%% file: main.tex
\DocumentMetadata{}
\documentclass[]{lipics-v2021}

\input{packages}
\input{macro}
\EventLogo{lipics-logo-bw}
\title{Efficient Dynamic Algorithms to Predict Short Races}

\author{Minjian Zhang}{University of Illinois Urbana-Champaign, USA}{minjian2@illinois.edu}{}{}
\author{Mahesh Viswanathan}{University of Illinois Urbana-Champaign, USA}{vmahesh@illinois.edu}{}{}

\authorrunning{M. Zhang and M. Viswanathan}

\ccsdesc[500]{Software and its engineering~Dynamic analysis}
\keywords{concurrency, prediction, data race, happens-before}
\Copyright{Anonymous Author(s)}
\begin{document}
\hideLIPIcs
\nolinenumbers
\maketitle
\input{abstract}

\input{intro}
\input{prelim}

\input{short_monitoring}
\input{windowhb}

\input{syncPhb}

\input{evaluation}

\input{relatedwork}

\input{conclusions}
\clearpage

\bibliographystyle{splncs04}
\bibliography{references}
\clearpage
\appendix

\input{appendix}

\end{document}

%% file: packages.tex
\usepackage{comment}
\usepackage{stmaryrd}
\usepackage{mathrsfs}
\usepackage{dsfont}
\usepackage{multicol}
\usepackage{doi}
\usepackage{multirow}
\usepackage{needspace}

\usepackage{tikz}
\usetikzlibrary{arrows.meta,positioning,fit,calc}

\usepackage[font=scriptsize]{caption}

\usepackage{svrsymbols}
\usepackage{xcolor}
\usepackage{algorithmic}
\usepackage{textcomp}
\usepackage[inline]{enumitem}
\usepackage{subcaption}
\usepackage{framed}
\usepackage{csquotes}

\usepackage[inline]{enumitem}

\usepackage{booktabs}

\usepackage[ruled,vlined,resetcount,linesnumbered]{algorithm2e}

\usepackage{hyperref}

\usepackage{comment}
\usepackage{todonotes}
\usepackage{longtable}

\usepackage{tikz}
\usetikzlibrary{arrows,automata,shapes,decorations,decorations.markings,calc, matrix,decorations.pathmorphing, patterns,backgrounds,shapes.misc,arrows.meta}



%% file: macro.tex
\newcommand{\figlabel}[1]{\label{fig:#1}}
\newcommand{\figref}[1]{Fig.~\ref{fig:#1}}
\newcommand{\seclabel}[1]{\label{sec:#1}}
\newcommand{\secref}[1]{Section~\ref{sec:#1}}

\newcommand{\probref}[1]{Problem~\ref{prob:#1}}
\newcommand{\thmlabel}[1]{\label{thm:#1}}
\newcommand{\thmref}[1]{Theorem~\ref{thm:#1}}
\newcommand{\proplabel}[1]{\label{prop:#1}}
\newcommand{\propref}[1]{Proposition~\ref{prop:#1}}

\newcommand{\tablabel}[1]{\label{tab:#1}}
\newcommand{\tabref}[1]{Table~\ref{tab:#1}}
\newcommand{\applabel}[1]{\label{app:#1}}

\newcommand{\algolabel}[1]{\label{alg:#1}}
\newcommand{\algoref}[1]{Algorithm~\ref{alg:#1}}

\newcommand{\rmv}[1] {}

\newcommand{\s}[1] {\mathsf{#1}}

\newcommand{\threads}[1]{\mathsf{Threads}_{#1}}
\newcommand{\locks}[1]{\mathsf{Locks}_{#1}}
\newcommand{\vars}[1]{\mathsf{Mem}_{#1}}

\newcommand{\match}[1]{\mathsf{match}_{#1}}
\newcommand{\TLClosure}[1]{\mathsf{TLClosure}_{#1}}

\newcommand{\SPIClosure}[1]{\mathsf{SPClosure}_{#1}}
\newcommand{\SPIdeal}[1]{\mathsf{SPIdeal}_{#1}}

\newcommand{\prev}[1]{\mathsf{prev}_{#1}}

\newcommand{\acq}{\mathtt{acq}}
\newcommand{\rel}{\mathtt{rel}}
\newcommand{\rd}{\mathtt{r}}
\newcommand{\wt}{\mathtt{w}}

\newcommand{\lk}{\ell}

\newcommand{\tr}{\sigma}

\newcommand{\ThreadOf}[1]{\mathsf{thr}(#1)}
\newcommand{\ind}{\mathsf{index}}

\newcommand{\racetype} {\s{R}}
\newcommand{\winsz} {\s{w}}
\newcommand{\trsz} {\s{N}}
\newcommand{\numthr} {\s{T}}
\newcommand{\numlk} {\s{L}}
\newcommand{\numvar} {\s{V}}
\newcommand{\numacq}{\s{A}}

\newcommand{\OpOf}[1]{\mathsf{op}(#1)}

\newcommand{\substr}[3] {#1[#2,#3]}

\newcommand{\setpred}[2]{\{#1 \,|\, #2\}}

\DeclareFontFamily{U}{mathb}{\hyphenchar\font45}
\DeclareFontShape{U}{mathb}{m}{n}{
<5> <6> <7> <8> <9> <10> gen * mathb
<10.95> mathb10 <12> <14.4> <17.28> <20.74> <24.88> mathb12
}{}
\DeclareSymbolFont{mathb}{U}{mathb}{m}{n}
\DeclareMathSymbol{\sqsubsetneq}    {3}{mathb}{"8A}

\newcommand{\acrtr}{\textsf{tr}\xspace}
\newcommand{\acrto}{\textsf{TO}\xspace}

\newcommand{\acrhb}{\textsf{HB}\xspace}

\newcommand{\ord}[2]{\leq^{#1}_{\mathsf{#2}}}

\newcommand{\strictord}[2]{<^{#1}_{\mathsf{#2}}}
\newcommand{\access}{\textsf{Acc}\xspace}
\newcommand{\cs}{\textsf{CS}\xspace}

\newcommand{\trord}[1]{\ord{#1}{\acrtr}}
\newcommand{\stricttrord}[1]{\strictord{#1}{\acrtr}}
\newcommand{\tho}[1]{\ord{#1}{\acrto}}
\newcommand{\stricttho}[1]{<^{#1}_{\acrto}}

\newcommand{\hb}[1]{\ord{#1}{\acrhb}}

\newcommand{\lw}[1]{\mathsf{lw}_{#1}}


\input{pseudocode_macros}

\renewcommand{\emptyset}{\varnothing}

\newcommand{\nats}{\mathbb{N}}
\newcommand{\rapid}{\textsf{RAPID}\xspace}

\newcommand{\fasttrack}{\textsc{FastTrack}\xspace}
\newcommand{\shortfasttrack}{\textsc{ShortFastTrack}\xspace}

\newcommand{\syncp}{\textsc{SyncP}\xspace}
\newcommand{\shortsyncp}{\textsc{ShortSyncP}\xspace}

\newcommand{\tsan}{\textsc{ThreadSanitizer}\xspace}

\newcommand{\ft}{\mathsf{FT}}	
\newcommand{\ltime}{\mathsf{L}}
\newcommand{\ctime}{\mathsf{C}}
\newcommand{\ltimeft}{{\ltime}_{\ft}}
\newcommand{\ctimeft}{{\ctime}_{\ft}}

\newcommand{\cle}{\sqsubseteq}
\newcommand{\CFT}{\mathbb{C}}
\newcommand{\window}{\mathbb{W}}

\newcommand{\RFT}{\mathbb{C}^{\rd}}

\newcommand{\epch}{\mathsf{e}}

\newcommand{\np} {\text{NP}}
\newcommand{\wone} {\text{W}[1]}

\newcommand{\activate}{\ensuremath{\mathsf{activate}}}
\newcommand{\collect}{\ensuremath{\mathsf{collect}}}
\newcommand{\spill}{\ensuremath{\mathsf{spill}}}
\newcommand{\replay}{\ensuremath{\mathsf{replay}}}
\newcommand{\Open}{\mathsf{Open}}
\newcommand{\spn}{\mathsf{span}}

\newcommand{\spanof}[2]{\spn(#1,#2)}

\newcommand{\myparagraph}[1] {\vspace*{0.1in}\noindent{\bf #1}\ }

\providecommand{\bot}{\perp}

\providecommand{\TLClosure}[1]{\mathsf{TLClosure}_{#1}}
\providecommand{\SPIClosure}[1]{\mathsf{SPClosure}_{#1}}
\providecommand{\SPIdeal}[1]{\mathsf{SPIdeal}_{#1}}
\providecommand{\prev}[1]{\mathsf{prev}_{#1}}



%% file: pseudocode_macros.tex
\SetKwProg{myfun}{function}{}{}
\SetKwProg{myhandler}{handler}{}{}
\SetKwProg{pohandler}{post-handler}{}{}

\SetKwFunction{init}{initialize}
\SetKwFunction{preaccess}{preHandler}

\SetKwFunction{resetClks}{resetClocksToNull}
\SetKwFunction{checkAndUpdate}{checkRaceAndUpdateMetadata}
\SetKwFunction{shouldSample}{shouldSample}
\SetKwFunction{newSample}{startNewSample}
\SetKwFunction{checkEndSample}{postHandler}

\SetKwFunction{endSample}{shouldEndSample}
\SetKwFunction{inSample}{inLocalSample}
\SetKwFunction{decidesampling}{decideSampling}
\SetKwFunction{skipevent}{returnWithoutInvokeEventHandler}

\SetKwFunction{merge}{mergeSubtraces}
\SetKwFunction{fixpoint}{ComputeSPIdeal}
\SetKwFunction{maxLB}{maxLowerBound}
\SetKwFunction{checkRace}{checkRace}
\SetKwFunction{rdhandler}{read}
\SetKwFunction{sdwindow}{windowSliding}
\SetKwFunction{wrapper}{windowWrapper}
\SetKwFunction{restore}{restoreAcquireInfo}
\SetKwFunction{store}{storeAcquireInfo}
\SetKwFunction{allocate}{allocation}
\SetKwFunction{deallocate}{deallocation}

\SetKwFunction{wthandler}{write}
\SetKwFunction{acqhandler}{acquire}
\SetKwFunction{relhandler}{release}
\SetKwInput{Params}{Parameters}
\SetKwInput{Input}{Input}
\SetKwInOut{Output}{Output}
\SetKw{Let}{let}
\SetKw{Break}{break}
\SetKw{NOT}{not}
\SetKw{declare}{declare}
\SetKw{check}{check}
\SetKw{exit}{exit}
\SetKwFor{Foreach}{for each}{}{}%
\DontPrintSemicolon
\newcommand{\mx}{\sqcup}

%% file: abstract.tex
\begin{abstract}
We introduce and study the problem of detecting short races in an observed trace. Specifically, for a race type $\racetype$, given a trace $\tr$ and window size $\winsz$, the task is to determine whether there exists $\racetype$-race $(e_1,e_2)$ in $\tr$ such that the subtrace starting with $e_1$ and ending with $e_2$ contains at most $\winsz$ events. We present a monitoring framework for short-race prediction and instantiate the framework for happens-before and sync-preserving races, yielding efficient detection algorithms. Our happens-before algorithm runs in the same time as {\fasttrack} but uses space that scales with $\log \winsz$ as opposed to $\log |\tr|$. For sync-preserving races, our algorithm runs faster and consumes significantly less space than {\syncp}. Our experiments validate the effectiveness of these short race detection algorithms: they run more efficiently, use less memory, and detect significantly more races under the same budget, offering a reasonable balance between resource usage and predictive power.
\end{abstract}

%% file: intro.tex
Developing reliable concurrent programs is a complex task, making formal reasoning principles essential for their analysis. The difficulty in designing and verifying concurrent systems stems from the inherent nondeterminism in program behavior caused by inter-process communication and scheduling. Managing this nondeterminism during development is challenging, often resulting in error-prone software. Furthermore, concurrency-related bugs are notoriously difficult to reproduce manually, highlighting the need for automated detection techniques to improve software development productivity.

A data race (or simply a race) occurs when a thread in a multi-threaded program accesses a shared memory location while another thread modifies it without proper synchronization. Data races often signal serious program defects~\cite{lpsz08} and have been associated with data corruption, compilation failures~\cite{Narayanasamy2007,boehmbenign2011,racemob}, and critical system failures~\cite{SoftwareErrors2009,evil2012}. As a result, extensive research has focused on detecting and preventing data races in multi-threaded programs.

Although static techniques for data race detection exist, dynamic race detection methods are more widely used due to their superior scalability and precision. The goal of dynamic race detection is to analyze an observed program trace to determine whether it contains evidence of a data race in the program that generated it. A key objective of this approach is to be \emph{predictive} --- not only identifying data races in the observed trace but also uncovering races in potential \emph{alternate} executions that the program could exhibit. This predictive capability improves the practical utility of dynamic race detection. However, the requirement to be predictive also makes the problem computationally challenging --- dynamic race detection is $\np$-complete and $\wone$-hard, rendering it fixed-parameter intractable~\cite{Mathur20,kmp21}.

The challenge of intractability has driven research into specialized types of predictive data races~\cite{lamport1978time,maxcausalmodels,cp2012,wcp2017,racechaser,shb2018,Roemer18,PavlogiannisPOPL20,Roemer20,SyncP2021,Shi2024}. These approaches aim to balance predictive power with computational tractability. Race definitions typically rely on either partial orders~\cite{lamport1978time,cp2012,wcp2017,shb2018} or constraints on the types of alternate executions considered~\cite{maxcausalmodels,PavlogiannisPOPL20,SyncP2021,Shi2024}.

Even if a race type is tractable, meaning it has a polynomial-time algorithm, the practical overhead can still be prohibitively high, limiting the size of traces that can be analyzed. Even for race detection algorithms that run in a single pass with linear-time streaming, lower bounds indicate that they require memory proportional to the trace size~\cite{wcp2017,SyncP2021}, which can be substantial. To address this, \emph{windowing}~\cite{cp2012,maxcausalmodels} has been proposed as a technique to reduce overhead in dynamic race detection. This approach divides the trace into subtraces of a fixed length (say) $\winsz$, and the algorithm searches for races only within each subtrace. By limiting analysis to bounded-length segments, windowing helps control computational overhead, making it feasible to analyze larger traces.
\textcolor{black}{Windowing imposes a key limitation on predictive analysis. Certifying a predictive race typically depends on context outside the two conflicting events and the events in between. If some key events in the context are outside the window, no algorithm can soundly certify the race, even though the two racy events themselves are in view.}
Therefore we focus on a problem that generalizes windowing. Rather than limiting race detection to predefined subtraces, we investigate whether a trace contains \emph{short} races --- races where the distance between the events involved is bounded. Our approach augments the window with a constant-space summary that captures exactly the out-of-window context required to certify predictive races.

Beyond reducing computational overhead, a key motivation for windowing and our method of detecting short races is that, in practice, racy programs frequently manifest short races. Furthermore, efficient algorithms for detecting such short races can also enhance randomized approaches to race detection~\cite{RPT2023}.
We present a monitoring framework for short-race detection on streaming traces: as the trace is processed online, the monitor maintains a fixed-size window as well as compact summary of the relevant history and updates it incrementally under window advances. We then instantiate the framework for two widely used race notions—happens-before and sync-preserving races. Happens-before races~\cite{lamport1978time} are the oldest and most widely used data race definition. It is defined using the happens-before partial order and underpins industrial tools like \tsan~\cite{threadsanitizer}. Sync-preserving races~\cite{SyncP2021} is a generalization of happens-before with greater predictive power. Both race types can be detected using linear-time, single-pass algorithms, though any linear-time algorithm for detecting sync-preserving races must use linear space~\cite{SyncP2021}. We explore whether these race definitions allow for more efficient detection of short races with reduced overhead. Our results are summarized in \tabref{overview}.

\vspace*{-0.1in}
\input{overview-table}

\vspace*{-0.3in}
Our algorithm for detecting short happens-before races of size $\winsz$ maintains a sliding window of the most recent $\winsz$ events and checks for races within this window. The correctness of this strategy hinges on the observation that whether a pair $(e_1, e_2)$ is ordered by happens-before depends only on the events between them. Our approach builds on \fasttrack~\cite{fasttrack}, the most efficient known HB detector, but modifies it to avoid the limitations of naïve windowing. Simply restricting the search to the last $\winsz$ events—while keeping standard vector-clock semantics intact—does not improve asymptotic space, since timestamps still scale with the full trace length $\trsz$ (e.g., the per-thread maximum number of lock operations), resulting in a $\log \trsz$ dependence. In contrast, our algorithm eliminates this dependence by re-anchoring time to the window and compacting metadata so that both the number and the magnitude of clock components are bounded by $\winsz$ (yielding $\log \winsz$ rather than $\log \trsz$). The algorithm matches {\fasttrack}’s running time but uses less space: whereas \fasttrack’s vector-clock components can grow with the trace length, ours are capped by $\winsz$ through careful bookkeeping and novel data structures. While these results are of theoretical significance, their practical impact may be modest given {\fasttrack}’s already strong scalability.

Detecting short sync-preserving races presents additional challenges that must be overcome. Unlike happens-before races, determining whether a pair $(e_1,e_2)$ forms a sync-preserving race depends not only on the events between $e_1$ and $e_2$ but also on prior memory accesses and synchronization events throughout the trace. This is the reason why the problem of detecting sync-preserving races has a linear space lower bound for any one-pass, streaming algorithm. Consequently, a simple sliding window of the last $\winsz$ events is insufficient for a sound detection algorithm. Our key insight is that soundness can be preserved by maintaining metadata from a select set of events outside the current window, in addition to tracking the sliding window itself. This observation has broader implications, suggesting the potential for efficient algorithms to detect short races across other race definitions. By leveraging carefully designed data structures, we develop an algorithm that detects short sync-preserving races soundly while running asymptotically faster and using significantly less space (see \tabref{overview}).

Our experimental results highlight the potential of short race detection algorithms. While the overhead reductions for happens-before races are not as pronounced as those for sync-preserving races, our experiments show that, in many cases, detecting short races leads to lower running times. Additionally, even with smaller windows, many of the races detected by {\fasttrack} are still reported, showcasing the strong predictive power of these algorithms. We believe short race detection algorithms are likely to be particular significant for sync-preserving races and this is borne out by our experiments. Our experiment shows that, when bounded to short race lengths, our algorithm yields significant improvements across nearly all aspects—including throughput, race prediction power, running time, and memory usage—thereby substantially improving scalability while preserving strong predictive power.

The rest of the paper is organized as follows. \secref{prelim} introduces basic notation, race definitions, and the problems we tackle in this paper. \secref{frame} introduces our monitoring framework for the problem. Instantiations for happens-before races are discussed in \secref{hb} and for sync-preserving races in \secref{sync-preserve}. Experiments are presented in \secref{experiments} before concluding in \secref{conc}.

%% file: overview-table.tex
\begin{table}[!htbp]
\scriptsize
\centering
\begin{tabular}{|l|l|l|}

\toprule

\multicolumn{1}{|c|}{\textbf{$\racetype$-race}} & \multicolumn{1}{c|}{\textbf{Race Detection}} & \multicolumn{1}{c|}{\textbf{Short Races}} \\

\hline
& & \\
Happens-before & Time: $O(\trsz\numthr)$ & Time: $O(\trsz\numthr)$ \\
 & Space: $O((\numthr+\numvar+\numlk)\numthr\log \trsz)$ & Space: $O(\winsz + (\min(\winsz, \numvar+\numlk)+\numthr)\numthr\log \winsz)$\\
 & & \\
\hline
& & \\
Sync-Preserving & Time: $O(\trsz\numthr^2 + \numacq\numvar\numthr^3)$ & Time: $O(\trsz\numthr(\min(\winsz,\numthr))+\numacq\numthr(\min(\winsz,\numvar\numthr^2)))$\\
 & Space: $O(\numvar\numlk\numthr^3+\trsz\numthr\log\trsz)$ & Space: $O(\min(\winsz,\numvar\numthr)\numlk\numthr^2+(\winsz+\numvar+\numlk)\numthr\log\trsz)$\\
& & \\
\hline

\end{tabular}
\vspace{10pt}
\caption{Overview of results based on the following parameters. $\trsz$ is the size of the trace, $\winsz$ is the length of the short race, $\numthr$ is the number of threads, $\numvar$ is the number of variables, $\numlk$ is the number of locks, and $\numacq$ is the number of acquire events in the trace.}
\tablabel{overview}
\end{table}

%% file: prelim.tex
\section{Preliminaries}
\seclabel{prelim}

\subsection{Programs traces}
\seclabel{traces}

An event $e$ in a concurrent program is represented as a tuple $(t,o)$ where $\ThreadOf{e}=t$ denotes the thread identifier responsible for the event and $\OpOf{e} = o$ signifies the type of operation. In this work, we limit the type of operation $o$ to the following:
(a) read/write access to memory location $x$, (i.e. $o = \rd(x)$ or $o = \wt(x)$), or (b) acquire/release of a lock $\lk$ (i.e. $o = \acq(\lk)$ or $o = \rel(\lk)$).
Traces are sequences of events. We use $E_{\tr}$ to denote the set of events in a trace $\tr$ and $E_{\tr}^t$ to denote the events that are carried out by a particular thread t.
We adopt the convention that the first event in the sequence has index $0$.
For a trace $\tr=e_0e_1e_2....e_{n-1}$ of length $n$, the $i$th event $e_i$ is also denoted as $\tr[i]$ and
its index is denoted as $\ind_{\tr}(e_i)=i$.  
 For an event $e$, we say that $e\in\tr$ if there exists an index $i$ such that $\tr[i]=e$.
 A \emph{subtrace} $\substr{\tr}{i}{j} = e_ie_{i+1}\cdots e_{j}$ is the subsequence of $\tr$ of length $j-i+1$ from index $i$ to index $j$.
When $j \leq i$, we adopt the convention that $\substr{\tr}{i}{j}$ is the empty sequence $\varepsilon$.

For every read event $e=\rd(x)$, let $\lw{\tr}(e)$ denote the last write event on $x$ that occurs before $e$ in $\tr$.  If such an event does not exist, $\lw{\tr}(e)=\bot$. For every acquire and release event, let $\match{\tr}(e)$ denote the pairing release/acquire. Similarly, if such an event does not exist, $\match{\tr}(e)=\bot$. 
  
A trace $w$ is said to be well formed if it respects lock semantics; i.e., no thread acquires a lock held by another thread and releases a lock it does not hold. A subtrace $w'$ is well formed if it is a subtrace of a well-formed trace. Every well-formed trace, by definition, is a well-formed subtrace. If the matching acquire of every release in a subtrace is present then it is a well-formed trace.

 For a pair of events $e_1,e_2$ in a subtrace $\tr$, we write $e_1\trord{\tr}e_2$ if $\ind_{\tr}(e_1)\leq\ind_{\tr}(e_2)$, namely $e_1$ occurs before $e_2$ in $\tr$. If also $\ThreadOf{e_1}=\ThreadOf{e_2}$, we denote it by $e_1\tho{\tr}e_2$. The corresponding irreflexive orders are, respectively, $\stricttrord{\tr}$ and $\stricttho{\tr}$. 
 Two memory access events $e_1,e_2$ that access the same memory location (say) $x$ are \emph{conflicting} if at least one of $e_1,e_2$ is a write and $\ThreadOf{e_1}\neq \ThreadOf{e_2}$.

\subsection{Data Races}
\seclabel{races}

A data race occurs when a thread accesses a shared memory location while another thread is modifying its contents, without proper synchronization. In predictive race detection, the goal is to determine if an observed trace witnesses a data race either in the trace or in an alternate ``reordering'' of the trace. This relies on identifying alternate traces that a program producing the observed trace has. Thus, to define whether a trace has a predictive race, we need the notion of \emph{correct reorderings} that captures a set of traces that must be exhibited by a program producing the observed behavior.

\myparagraph{Correct reordering }{
For a given well-formed trace $\tr$, a well-formed trace $\tr'$ is said to be a \emph{correct reordering} of $\tr$ if (1) $E_{\tr'}\subseteq E_{\tr}$, (2) for any pair of events $e_1, e_2$, if $e_1\tho{\tr} e_2$ and $e_2\in E_{\tr'}$ then $e_1 \in E_{\tr'}$ and $e_1\tho{\tr'} e_2$, and (3) for every read event $e \in E_{\tr'}$, $\lw{\tr}(e)=\lw{\tr'}(e)$.
Correct orderings identify alternate executions of the program that generated the trace $\tr$. A correct reordering $\tr'$ of $\tr$ has a consistent thread order and ensures that each read event observes the effect of the same write event. It is worth observing that any trace $\tr$ is a correct reordering of itself.
}

\myparagraph{Predictive Races}{
The problem of dynamic race prediction asks if there exists a correct reordering of a given trace that witnesses a race. 

To define the presence of a race in a trace, we need the notion of events being enabled. An event $e$ of a well-formed trace $\tr$ is said to be \emph{enabled} in a correct reordering $\tr'$ of $\tr$ if $e \not\in E_{\tr'}$ and for every $e' \tho{\tr} e$, we have $e'\in E_{\tr'}$. Then a pair of conflicting events $e_1,e_2 \in E_{\tr}$ is said to be a \emph{predictive data race} in $\tr$ if there is a correct reordering $\tr'$ of $\tr$ such that $e_1$ and $e_2$ are enabled in $\tr'$.   
}

The problem of predictive race detection is known to be $\np$-complete and fixed-parameter intractable~\cite{Mathur20}. This has prompted many innovative approaches to identify specific classes of predictive races to achieve a balance between efficiency and predictive power. They are broadly classified into those based on partial orders~\cite{fasttrack,shb2018,wcp2017} and those based on restricting the class of correct reorderings~\cite{SyncP2021,Shi2024,PavlogiannisPOPL20}. We look at one important example from each category.  

\myparagraph{Happens-before Data Races}{
The \emph{happens-before partial order} $\hb{\tr}$ of a trace $\tr$ is the smallest partial order on $E_{\tr}$ such that for every pair of events $(e_1,e_2)$:
\begin{enumerate}
    \item if $e_1\tho{\tr}e_2$, then $e_1\hb{\tr}e_2$, and 
    \item if $e_1\trord{\tr}e_2$, $e_1=\rel(\lk)$ and $e_2=\acq(\lk)$ for some lock $\lk$ then $e_1\hb{\tr}e_2$.
\end{enumerate}
A subtrace $\tr$ is said to have a \emph{happens-before race} if there exists a pair of conflicting events $e_1$ and $e_2$ such that neither $e_1\hb{\tr}e_2$ nor $e_2\hb{\tr}e_1$.
}

Happens-before races are \emph{sound races}, i.e. if a trace $\tr$ has a happens-before race, then it has a predictive race. It is, however, important to highlight a subtlety --- even though happens-before races are sound and demonstrate the presence of a predictive race, not every happens-before race is a predictive race. Happens-before data races are one of the most popular types of data races and is the basis of widely used industrial tools such as \tsan~\cite{threadsanitizer}. It has an efficient vector clock based algorithm that is one of the reasons for its popularity. Let us introduce vector clocks and vector timestamps that play an important role in many dynamic data race detection algorithms.

\myparagraph{Vector timestamps and Clocks}{
Vector timestamps and vector clocks are essential concepts widely used to capture and reason about causal relationships between events in distributed or multi-threaded environments. 

A \emph{vector timestamp} is a mapping from $\threads{\tr}$ to $\nats$ that assigns each thread a natural number. It is typically represented as a vector of length $\numthr$, where $\numthr = |\threads{\tr}|$ is the number of threads in $\tr$. In race detection algorithms, vector timestamps are often associated with events with the vector timestamp $V_e$ of an event $e$ representing the downward closed set of events with respect to $\tho{\tr}$ that causally precede $e$; here, each entry $V_e(t)$ stores the index of the last event of thread $t$ that belongs to the set represented by $V_e$. For example, a vector timestamp is an efficient way to represent the set of all events that happen-before another event $e$. The \emph{join} of two vector timestamps $V_1$ and $V_2$ is denoted as $V_1\mx V_2$, and is the pair-wise maximum of the timestamps $\lambda t \max(V_1(t),V_2(t))$. The pointwise order on vector timestamps defines a partial order. Thus, $V_1\cle V_2$ if $\forall t (V_1(t)\leq V_2(t))$. We use $\bot$ to denote the default vector timestamp where every entry has the value $-1$.

A \emph{vector clock} is a variable that stores a vector timestamp. Vector clocks are used by race detection algorithms to store the timestamp of important events as the trace is observed. They play a crucial role in the algorithms for dynamic race detection.
}

The second type of data race that we will focus on in this paper is \emph{sync-preserving data races} which generalize happens-before races. They are obtained restricting attention to sync-preserving correct reorderings.

\myparagraph{Sync-preserving Races}{
A correct reordering $\tr'$ of $\tr$ is said to be a \emph{sync-preserving correct reordering} of $\tr$ if for every two acquire events from $E_{\tr'}$ on the same lock, i.e. $e_1=\acq(\lk)$ and $e_2=\acq(\lk)$, we have $e_1\trord{\tr'}e_2$ if and only if $e_1\trord{\tr}e_2$.
A pair of conflicting events $e_1,e_2$ is said to be a \emph{sync-preserving race} in a well-formed trace $\tr$ if there is a sync-preserving correct ordering $\tr'$ where both $e_1$ and $e_2$ are enabled. By definition, every sync-preserving race is a predictive race.
}

\begin{proposition}[\cite{SyncP2021}]
\proplabel{hb-syncp}
If a well-formed trace $\tr$ has a happens-before race then $\tr$ also has a sync-preserving race.
\end{proposition}

We now present an alternate characterization of sync-preserving races in terms of ideals or downward closed sets. This characterization, coupled with the use of vector clocks to represent downward-closed sets, forms the basis of the linear-time algorithm to detect sync-preserving races~\cite{SyncP2021}.

\myparagraph{Sync-preserving Closure}{ 
%
For a given well-formed trace $\tr$, a set of events $S\subseteq E_{\tr}$ is \emph{thread order and last-write closed  or TL-closed} if 
\begin{enumerate}
    \item $S$ is downward closed with respect to $\tho{\tr}$. 
    \item For every read event $e\in S$, if $\lw{\tr} (e)\neq \bot$, $\lw{\tr}(e)\in S$.
\end{enumerate}
The TL closure of a set $S$, denoted $\TLClosure{\tr}(S)$, is the smallest TL-closed set that contains $S$.  

Further, $S$ is said to be \emph{sync-preserving closed} if 
\begin{enumerate}
\item $S$ is TL-closed, and
\item for any pair of acquire events $a_1,a_2\in S$, if $a_1$, $a_2$ access the same lock and $a_1\stricttrord{\tr} a_2$, then $\match{\tr}(a_1)\in S$.
\end{enumerate} 
The sync-preserving closure of a set $S$, denoted $\SPIClosure{\tr}(S)$, is the smallest sync-preserving closed set containing $S$.

Let $\prev{\tr}(e)$ denote the last event in $\tr$ performed by thread $\ThreadOf{e}$ before $e$. Then, for an event $e \in E_{\tr}$, every sync-preserving correct ordering of $\tr$ must include all events in $\SPIClosure{\tr}({\prev{\tr}(e)})$ to ensure that $e$ is enabled. Therefore, two conflicting events $e_1$ and $e_2$ are in race if and only if $\SPIClosure{\tr}({\prev{\tr}(e_1),\prev{\tr}(e_2)})$ does not contain $e_1$ nor $e_2$~\cite{SyncP2021}; for events $f_1,f_2$, we use $\SPIClosure{\tr}(f_1,f_2)$ to denote the set $\SPIClosure{\tr}(\{f_1,f_2\})$. We will use $\SPIdeal{\tr}(e_1,e_2)$ to denote the set $\SPIClosure{\tr}({\prev{\tr}(e_1),\prev{\tr}(e_2)})$.
}

\subsection{Predictive Race Detection}
\seclabel{race-detection}

The predictive race detection problem is to check if an observed trace $\tr$ has a predictive race. As noted above, the problem of predictive race detection is known to be intractable and fixed-parameter intractable~\cite{Mathur20}. Therefore, research has focused on developing algorithms to check the existence of special types of predictive races, like happens-before or sync-preserving races. Let $\racetype$ be a type of predictive race. The $\racetype$-race detection problem is as follows.

\myparagraph{$\racetype$-race Detection Problem}{
Given a trace $\tr$, determine if there is a pair of conflicting events $e_1,e_2 \in E_{\tr}$ such that $(e_1,e_2)$ is a $\racetype$-race.
}

For both happens-before and sync-preserving, the race detection problem has one-pass linear-time algorithms. Though both automata-based~\cite{elmas2007goldilocks,ang2024predictivemonitoringstrongtrace} and vector clock-based~\cite{djit,fasttrack,SyncP2021} algorithms exist for both these types of races, the algorithms that scale to large traces in practice are the vector clock algorithms. Details of the vector clock algorithms for happens-before and sync-preserving will be presented in Sections~\ref{sec:fasttrack} and~\ref{sec:syncp}. Even if a race detection problem has a linear time algorithm, in many cases, the overhead may still be high, especially when processing very large traces. For example, any one-pass streaming algorithm that detects the presence of sync-preserving races must use linear space (i.e., space proportional to the length of the trace)~\cite{SyncP2021}. 

\emph{Windowing}~\cite{cp2012,maxcausalmodels} has been proposed as a possible approach to address this challenge. In windowing, the trace is divided into subtraces of length (say) $\winsz$, and the algorithm then only searches for races within each subtrace. The restriction to bounded length subtraces ensures that the overhead of the race detection algorithm is bounded.

In this paper, we generalize the windowing approach by asking whether a trace contains short races—races where the distance between the conflicting events is bounded—rather than restricting attention to fixed subtraces. In practice, racy programs often exhibit short races, and detecting them also underlies randomized algorithms for race detection~\cite{RPT2023}.

\myparagraph{Short $\racetype$-race Detection Problem}{
To formally reason about short races, we first define the notion of a span:
For indices $0\le i \le j < |\tr|$, define the span of $(i,j)$ as the length of the inclusive subtrace $\substr{\tr}{i}{j}$. Thus, 
$
   \spanof{i}{j} \triangleq j-i+1.
$
In particular, $\spanof{0}{1}=2$.
 
Then,  the short $\racetype$-race detection problem asks: given a trace $\tr$ and $\winsz \in \nats$, decide whether there exist conflicting events $e_i,e_j \in E_{\tr}$, with $i<j$, such that $e_i$ and $e_j$ form a $\racetype$-race and $\spanof{i}{j}\leq \winsz$.

}

\begin{remark}[Decision vs.\ reporting]
Although the problem above is stated in decision form, our algorithms target a stronger \emph{reporting} objective: as they scan the trace $\tr$ online, they identify and may report Short-$\racetype$ races witnessed by conflicting pairs within distance~$\winsz$. The decision problem is obtained as a special case by stopping after the first reported race (or, equivalently, by checking whether the reported set is empty). We adopt the reporting objective since dynamic race detectors (including happens-before detectors) are typically run over the entire trace and configured to emit races of interest.
\end{remark}

%% file: short_monitoring.tex
\section{Monitoring Framework for Short Races}
\seclabel{frame}
\input{example1}

A \emph{Short-$\racetype$} race is a $\racetype$-race witnessed by a conflicting pair $(e_i,e_j)$ with
$\spanof{i}{j}\leq \winsz$ in the full trace $\tr$.
We use a sliding window of size $\winsz$ purely as an \emph{online enumeration mechanism}:
any pair with $\spanof{i}{j}\leq \winsz$ must co-occur in the window, so restricting attention to
pairs that appear together in the window does not miss any Short-$\racetype$ witness.
Crucially, although the \emph{witness pair} is local, the predicate “$(e_i,e_j)$ forms a $\racetype$-race”
may depend on execution context outside the window. The challenge is therefore to retain just enough
\emph{out-of-window context} to evaluate the same full-trace $\racetype$ judgment for all future
window-local candidates, while keeping memory bounded by a function of $\winsz$.

We propose a generic monitoring framework that runs a sound-and-complete streaming detector
$A_{\racetype}$ for $\racetype$-races over $\tr$ under bounded memory.
As events arrive, the framework maintains a window of the last $\winsz$ events, and runs the underlying detector
$A_{\racetype}$ while dynamically updating and clearing its history-dependent state to enforce bounded space.
Abstractly, the window serves only to enumerate candidate witness pairs with $\spanof{i}{j}\le \winsz$, whereas
the detector state is pruned on-the-fly so that information irrelevant to future Short-$\racetype$ checks is discarded
as soon as it becomes obsolete.

For each incoming event $e$, the framework performs two operations:

\smallskip
\noindent\textbf{(1) $\activate(e)$.}
Insert $e$ into the window and advance the underlying detector state as if processing the full trace.
In particular, $\activate$ executes $A_{\racetype}$'s usual update logic and checks races involving $e$
against currently live in-window events (i.e., those within span at most $\winsz$).

\smallskip
\noindent\textbf{(2) $\collect(e_{\mathit{old}})$.}
When the window exceeds size $\winsz$, evict the oldest event $e_{\mathit{old}}$ and reclaim all detector
state that is \emph{no longer relevant} for any future Short-$\racetype$ witness.
Upon eviction, $\collect$ prunes the history-dependent state maintained by $A_{\racetype}$:
it discards components that can no longer affect any future Short-$\racetype$ check, while preserving the remaining portion of $A_{\racetype}$'s state as the necessary
out-of-window context.

\smallskip
The framework targets the reporting formulation of short-race detection and is required to be
\emph{sound} and \emph{complete} for Short-$\racetype$ races:
it never reports a false Short-$\racetype$ race, and it reports every Short-$\racetype$ race witnessed by
a conflicting pair $(e_i,e_j)$ with $\spanof{i}{j}\leq \winsz$.

Figure~\ref{fig:framework-mechanics} visualizes the high-level structure of our monitoring framework for short-race detection.
As events are processed online, we maintain a sliding window of size $\winsz$ to enumerate short-witness candidates,
while running the underlying streaming race detector $A$ and maintaining its internal state $M_A$.
The two operations correspond to the two window boundaries: \activate{} updates $M_A$ when a new event enters from the right,
whereas \collect{} evicts the oldest event on the left and prunes the parts of $M_A$ that cannot affect any future in-window check.
The remaining residual state serves as the necessary out-of-window context for subsequent candidates.

Two aspects make instantiation non-trivial.
First, one must characterize \emph{which} parts of $A_{\racetype}$'s state can be safely discarded when events leave
the window, and \emph{which} parts must survive as compact summaries to preserve full-trace judgments for all future
window-local candidates; otherwise, space would again scale with $\trsz$.
Second, one must implement $\collect$ so that enforcing bounded memory does not asymptotically increase the running time
of $A_{\racetype}$.
Our instantiations for happens-before and sync-preserving races address both challenges by deriving race-specific compact
summaries and by integrating boundary maintenance into $A_{\racetype}$'s existing update logic without changing its
asymptotic time complexity.

In the remainder of the paper, we instantiate this framework for both happens-before races and sync-preserving races
(via \syncp{}), obtaining detectors whose space usage depends on $\winsz$ rather than $\trsz$.
These results are summarized in \tabref{overview}, and our experiments demonstrate their effectiveness.

%% file: example1.tex
\begin{figure}[t]
\centering
\begin{tikzpicture}[
  font=\small,
  >=Stealth,
  stream/.style={->, line width=1.05pt, draw=black!80},
  dot/.style={circle, draw=black!80, line width=0.9pt, minimum size=7.2mm, inner sep=0pt},
  pastdot/.style={dot, fill=white},
  olddot/.style={dot, fill=orange!75},
  newdot/.style={dot, fill=cyan!55},
  win/.style={draw=green!50!black, fill=green!8, rounded corners=2pt, line width=0.7pt},
  opbox/.style={
    draw=black!65, fill=white, rounded corners=2pt, line width=0.65pt,
    align=center, inner xsep=12pt, inner ysep=7pt, minimum height=10mm,
    minimum width=3.65cm
  },
  arr/.style={->, line width=0.9pt, draw=black!75},
  tiny/.style={font=\scriptsize, text=black!75},
]

\node at (6.60,2.05) {\bfseries Streaming trace $\sigma$};

\draw[stream] (0.2,1.25) -- (12.9,1.25);

\node[pastdot] at (1.1,1.25) {};
\node[pastdot] at (2.0,1.25) {};
\node[pastdot] at (2.9,1.25) {};

\node[pastdot] at (10.6,1.25) {};
\node[pastdot] at (11.5,1.25) {};

\draw[win] (4.05,0.70) rectangle (9.65,1.80);

\node[olddot] (eold) at (4.60,1.25) {};
\node[pastdot]        at (5.90,1.25) {};
\node[pastdot]        at (7.15,1.25) {};
\node[pastdot]        at (8.30,1.25) {};
\node[newdot] (enew) at (9.05,1.25) {};

\node[opbox] (col) at (4.95,-0.05)
  {{\bfseries Collect}\\[-2pt]\tiny evict $e_{\mathit{old}}$;\ prune};

\node[opbox] (act) at (8.55,-0.05)
  {{\bfseries Activate}\\[-2pt]\tiny ingest $e_{\mathit{new}}$;\ update $M_A$};

\coordinate (bendL) at (4.60,0.98);
\coordinate (bendR) at (9.05,0.98);

\draw[arr] (eold.south) -- (bendL) |- (col.north);
\draw[arr] (enew.south) -- (bendR) |- (act.north);

\coordinate (stSW) at (1.45,-2.80);
\coordinate (stNE) at (11.95,-1.55);
\draw[draw=black!65, line width=0.65pt, rounded corners=2pt, fill=black!1] (stSW) rectangle (stNE);

\node at (6.70,-1.38) {\bfseries Detector state $M_A$};

\path[fill=black!2, draw=black!35, line width=0.5pt, dashed, rounded corners=1.5pt]
  (1.90,-2.55) rectangle (6.35,-1.80);
\node[tiny] at (4.13,-2.18) {pruned / discarded};

\path[fill=cyan!10, draw=black!25, line width=0.5pt, rounded corners=1.5pt]
  (6.45,-2.55) rectangle (11.50,-1.80);
\node[tiny] at (8.98,-2.18) {retained (relevant) state};

\coordinate (bendCol) at (4.95,-0.72);
\coordinate (bendAct) at (8.55,-0.72);
\coordinate (stTopL)  at (5.25,-1.55);
\coordinate (stTopR)  at (8.15,-1.55);

\draw[arr] (col.south) -- (bendCol) -| (stTopL);
\draw[arr] (act.south) -- (bendAct) -| (stTopR);

\end{tikzpicture}
\caption{High-level mechanics of the monitoring framework.
A window of size \(\winsz\) slides over the event stream to enumerate candidates.
When the window advances, \collect{} (left boundary) evicts \(e_{\mathit{old}}\) and prunes history-dependent state
irrelevant to future Short-\(\racetype\) checks.
\activate{} (right boundary) ingests \(e_{\mathit{new}}\) and updates the underlying detector state \(M_A\),
keeping memory bounded by \(O(\winsz)\).}
\label{fig:framework-mechanics}
\end{figure}
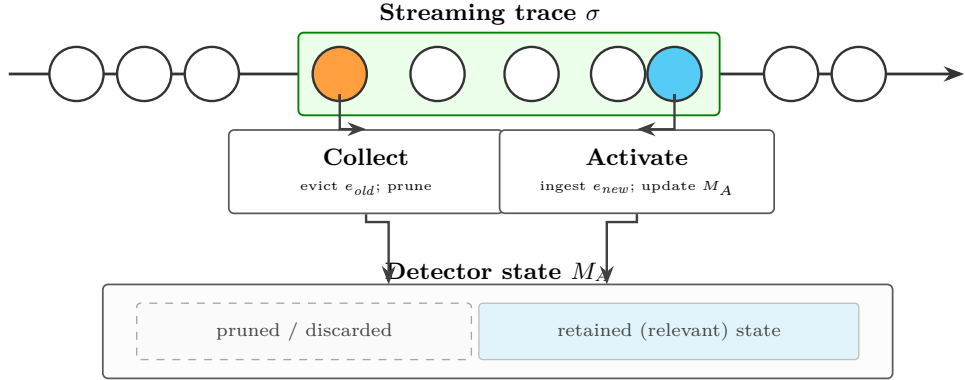

%% file: windowhb.tex
\section{Happens-before Races}
\seclabel{hb}
Happens-before (HB) races are the oldest and most widely used notion of data race. Within the framework of predictive analysis, an HB race can be regarded as the simplest predictive race. HB also remains the standard notion for concurrency testing—where one examines a concrete execution for the presence of a race—and it underpins widely deployed tools such as ThreadSanitizer (TSan)~\cite{threadsanitizer}. A classic and efficient algorithm for sound HB race detection is {\fasttrack}~\cite{fasttrack}, which computes the HB relation using vector clocks. {\fasttrack} is already highly efficient, achieving linear time and logarithmic space—performance. Nonetheless, studying short HB races is theoretically interesting: because HB is the dominant definition of race, clarifying its complexity under restrictions such as short witnesses provides a principled baseline and sharpens our understanding of locality constraints in race detection.

In practice, there are precedents for windowing-style approaches for HB-races that aim to keep {\fasttrack} tractable. For instance, tools like TSan maintain a small, per-location bounded history (“window”) of recent accesses and report a race only if the conflicting pair falls within that window, thereby limiting per-access metadata checks. However, these techniques are primarily heuristic in nature: they do not change the asymptotic complexity of HB race detection, nor are they formulated as algorithms parameterized by an explicit resource bound. Motivated by this observation, we study the \emph{short HB-race} problem from a theoretical perspective and design an algorithm whose complexity is parameterized by the window size~$\winsz$. Our algorithm is sound and complete for all HB-races whose witnesses are at most $\winsz$ events apart, making explicit a principled cost–coverage trade-off that is only implicit in existing tools.

We first note that happens-before (HB) races satisfy a locality property: for any two events in a trace $\tr$, whether they form an HB-race depends only on the events that occur between them.

\begin{proposition}[Locality of HB races]
\proplabel{hb-locality}
Let $\tr = e_0 e_1 \dots e_n$ be a well-formed trace, and suppose it contains a conflicting pair $(e_i, e_j)$ that is an HB-race in $\tr$. Then $(e_i, e_j)$ is also an HB-race in the subtrace $\tr[i:j]$.
\end{proposition}

This observation suggests that the short HB-race detection problem can be solved via a sliding-window algorithm that searches for HB-races within subtraces of length $\winsz$.

We therefore begin by recalling the {\fasttrack} algorithm, which we will later adapt into our framework to reduce asymptotic space requirements while still detecting short HB-races. Throughout, we fix the following parameters of the trace under analysis: $\trsz$ for trace length, $\numthr$ for the number of threads, $\numvar$ for the number of variables, $\numlk$ for the number of locks, and $\winsz$ for the maximum length of short races considered in \secref{hb-short}.

\rmv{
The solutions we propose to address \probref{prob1} are inspired by classic sliding window algorithms. For a given race definition and a budget $k$, we modify the existing optimal event-based algorithm for that race definition to maintain a sliding window of size $k$ that actively keeps the last $k$ events seen so far. Recall that event-based approaches operate by associating specific information with each event, which is then accessed by future events. The key intuition is that when an event falls outside the window, we discard the associated information. This not only frees up memory but also ensures that subsequent events cannot attempt to declare a race involving the removed event. 
There are two potential challenges we must address for this approach.  (a) For certain race definitions, it may be necessary to analyze more than $k$ events to determine whether two conflicting events, separated by fewer than $k$ events, constitute a race. Namely, we need to ensure that discarding events outside the window does not lead to unsoundness.(b) We also need to efficiently maintain the dynamic association between each piece of information or object and its corresponding event.
In this section, we present our first result: an adaptation of the \fasttrack algorithm to solve \probref{prob1}, ensuring it maintains the same time complexity while eliminates its space complexity dependence on the trace size. 
}

\subsection{The {\fasttrack} Algorithm}
\seclabel{fasttrack}

The \fasttrack algorithm computes a vector timestamp for each event. Define \emph{local time} of an event $e$ as the number of events that have occurred in $\tr$ before $e$ within the same thread as $e$.
\begin{align}
    \ltimeft^\tr(e) = |\setpred{f}{f \stricttho{\tr} e}|
\end{align}
The \emph{causal time stamp} of an event $e$ is subsequently defined as     
\begin{align}
    \ctimeft^\tr(e) = \lambda t \cdot \max \setpred{\ltimeft^\tr(f)}{\ThreadOf{f} = t, f \hb{\tr} e}
\end{align}
Intuitively, the local time ensures that every event has different times. The key property of the causal time stamp is that $e_1\hb{\tr}e_2$ if and only if $\ctimeft^\tr(e_1)\cle \ctimeft^\tr(e_2)$. 
\input{algo-ft}

\fasttrack (\algoref{djitp}) computes $\ctimeft^\tr(e)$ for each event dynamically as the trace unfolds, avoiding the need for a vector clock for every event. Instead, it maintains only a few vector clocks and variables. For each variable $x$, it maintains an epoch $\epch^w_x$ that stores the local time of the last write to $x$; in particular, an epoch has the form $t@c$, where $t$ is a thread identifier and $c$ is its local time. Given a vector clock $\CFT$, the notation $\epch^{w}_x \cle \CFT$ holds iff $c \le \CFT[t]$.

In addition, it has vector clocks that store the local time of each thread's last read on each variable ($\RFT_x$), the causal time of the last release of every lock ($\CFT_\lk$) and the causal time of the last event from each thread ($\CFT_t$). 

\begin{remark}
    The \fasttrack algorithm is typically presented with an optimization where the local clock is incremented only after release events. This optimization does not change the asymptotic complexity of the algorithm. 
\end{remark}

\myparagraph{Analysis}{
\fasttrack allocates $\numthr+\numvar+\numlk$ vector clocks. Since each vector clock contains $\numthr$ entries, and the value of each entry can be as large as the trace length $\trsz$, the algorithm uses $O((\numthr+\numvar+\numlk)\numthr \log \trsz)$ space. For each read or write event, \fasttrack compares the timestamp of the current event with the timestamp of the last access and reports races if necessary. For each acquire and release event, it performs vector join and copy operations to propagate causal information. As a result, the algorithm takes $O(\trsz\numthr)$ time, with each operation (release, acquire, and write) incurring a cost of $O(\numthr)$.
}

\subsection{Detecting Short Happens-before Races}
\seclabel{hb-short}
We now instantiate our monitoring framework for Short-HB races by taking {\fasttrack} as the underlying streaming detector \(A_{\mathrm{HB}}\). By the locality of HB races (Prop.~\ref{prop:hb-locality}), whether a conflicting pair \((e_i,e_j)\) with \(j-i\le \winsz\) forms an HB-race depends only on the subtrace \(\tr[i\!:\!j]\). Consequently, the sliding window of size \(\winsz\) is therefore used purely to enumerate candidate pairs. 

The remaining challenge is purely algorithmic: how should we design the
\(\activate(\cdot)\) and \(\collect(\cdot)\) operators so that the framework improves
both time and space complexity?
A naive approach is to maintain a window of the last \(\winsz\) events and
re-run (or re-simulate) \(A_{\mathrm{HB}}\) from scratch on the window contents after each shift.
Although this reduces the \emph{space} footprint relative to running \(A_{\mathrm{HB}}\) on the full trace,
it introduces a prohibitive \emph{runtime} overhead (essentially an extra factor of \(\Theta(\winsz)\)).

Our solution instead designs \(\activate(\cdot)\) to process each incoming event using the standard
{\fasttrack} update rules, while maintaining the invariant that the detector state only reflects
events currently in the window.
Accordingly, the main technical difficulty lies in \(\collect(\cdot)\):
as events expire from the window, we reclaim event-associated access and synchronization metadata
and update any remaining state that may still reference it, so that the detector's space usage depends
on \(\winsz\) rather than \(\trsz\), without increasing the asymptotic running time of {\fasttrack}.

Abstractly, {\fasttrack} maintains a state that summarizes key events observed so far.
In particular, for each memory location, it stores metadata for the most recent read and write
accesses (per thread) in that location's clock state; and for each thread, it stores in the
thread's vector clock the most recent release from every other thread that the thread has
observed so far.

In our framework for short races, \(\activate(e)\) allocates (or reuses) the metadata record associated with \(e\)
and installs it into the {\fasttrack} state, so that subsequent events in the window can reference it.
Dually, when \(e\) expires, \(\collect(e)\) must (i) remove all references to \(e\)'s record from the current
detector state and (ii) reclaim the record, while doing so in time proportional to the number of \emph{actual}
references.

For read and write events, this is straightforward: their records are referenced only via the per-location
``last-access'' fields, so \(\collect(e)\) can clear the corresponding entry if it still points to \(e\),
and then reclaim the record.

Release events are more subtle. A single release record can be propagated into many threads' vector clocks,
so by the time the release leaves the window, it may be referenced by a large number of vector-clock entries.
A naive \(\collect(e)\) would therefore have to scan all threads' clocks to find and clear such references.
To avoid this cost, we equip \(\activate(\cdot)\) with a transposed bookkeeping structure:
whenever processing an event causes some vector-clock entry to start referencing a release record \(r\),
we additionally register that fact in an index for \(r\).
As a result, when \(\collect(r)\) is invoked at the window boundary, it can enumerate exactly the vector-clock
entries that currently reference \(r\), update them in place, and then safely reclaim \(r\)'s metadata,
without any global traversal.

The discussion above gives the high-level design of \(\activate(\cdot)\) and \(\collect(\cdot)\).
At that level, however, two gaps remain.
First, we have not yet argued \emph{soundness}: after we reclaim metadata for expired events, the maintained state
must still be equivalent to running {\fasttrack} on the current window, so that every reported short race is a
genuine HB-race.
Second, we have not yet achieved the desired \emph{space} bound: even if we store metadata only for window-resident
events, {\fasttrack} still uses \emph{global} local times whose bit-width grows with \(\trsz\).

\myparagraph{Soundness via window-equivalence.}
For HB races, locality (\propref{hb-locality}) means that reporting short races depends only on the current
window, hence our correctness argument reduces to maintaining a \emph{window-equivalence} invariant:
after each update, the maintained state is semantically equivalent (up to renaming of window-relative timestamps) to
running {\fasttrack} on the current window alone.
The only nontrivial aspect is ensuring that \(\collect(\cdot)\) can reclaim window-expired metadata without breaking
this invariant; this is precisely what the transposed bookkeeping mechanism described above guarantees, while avoiding
any global traversal and thus preserving {\fasttrack}'s asymptotic update time.

\myparagraph{Eliminating the residual \(\log \trsz\) dependence.}
We now turn to the complementary \emph{space} goal.
In vanilla {\fasttrack}, each vector-clock component stores a global local time that can grow up to \(\trsz\),
requiring \(\Theta(\log \trsz)\) bits per component.
Since Short-HB races are determined entirely by the current window, we can instead use a window-relative notion of
time and recycle identifiers as the window advances, reducing this cost to \(\Theta(\log \winsz)\) bits per component.
More specifically, we replace the timestamp assigned to each events by their index in the window. To effectively maintain the window, we implement it as a circular array. Let the window indices be $\{0,1,\ldots,\winsz-1\}$, we represent the window as a circular array $W[0..\winsz-1]$ with head $h\in\{0,1,\ldots,\winsz-1\}$.

Formally, let \(\window = (W, h)\) be a window of length \(\winsz\), where
\(h \in \{0,1,\ldots,\winsz-1\}\) is the head position.
For \(i,j \in \{0,1,\ldots,\winsz-1\}\), define \(i \leq_\window j\) if and only if either
(a) \(j=h\), or (b) \(j < h < i\), or (c) \(h < \min(i,j)\) or \(h > \max(i,j)\) and \(i < j\).
The ordering between vector timestamps (\(\cle_\window\)) is the usual pointwise ordering where the ordering on
components is \(\leq_\window\).

All these ideas are incorporated in our algorithm which is shown as \algoref{algo2}. We highlight the notable changes and subtle differences of our algorithm when compared with {\fasttrack} (\algoref{djitp}). The algorithm uses a circular array $\window$ of size $\winsz$ to represent the window. The function \activate{} takes in a new event $e$ and updates $h$. It first invokes the \collect{} function on the event that is being dropped from the window and then calls the corresponding \textbf{handler} on the new event before it adds $e$ to the window. $S_i$ represents the set of threads and locks whose entries are associated with the release event stored at $W[i]$, as discussed above. Intuitively, the \textbf{post-handler} reverts what have happened in \textbf{handler} and, as such, does not affect the asymptotic running time of the algorithm.
\begin{theorem}
    \thmlabel{algo2-correctness}
    \algoref{algo2} correctly solves the short happens-before-race detection problem. It runs in time $O(\trsz\numthr)$ and uses space $O(\winsz + (\min(\winsz,\numvar+\numlk)+\numthr)\numthr\log \winsz)$. 
\end{theorem}

A proof of \thmref{algo2-correctness} is presented in Appendix in the replication package.

\begin{example}[Circular window on a short trace]
\label{ex:circular-window-suffix}
Consider the trace
\[
e_1=[t_1:\wt(y)],\;
e_2=[t_1:\wt(x)],\;
e_3=[t_2:\wt(x)],\;
e_4=[t_2:\wt(y)].
\]
There are two HB-races: $(e_1,e_4)$ with span~4 and $(e_2,e_3)$ with span~2.
Let the window size be $\winsz=2$ with circular head $h\in\{0,1\}$ initialized to $-1$.

\smallskip
Processing $e_1,e_2$ fills the window without eviction, recording
$\epch^w_y=t_1@0$ and $\epch^w_x=t_1@1$. 
When $e_3$ arrives, $e_1$ is evicted and collected, resetting
$\epch^w_y$ to null.
Since there is no happens-before edge between $e_2$ and $e_3$,
the FastTrack write check fails and the detector reports the HB-race
$(e_2,e_3)$.

\smallskip
Next, processing $e_4$ evicts $e_2$.
Because $\epch^w_x$ no longer refers to the evicted write,
no additional race is detected, and $(e_1,e_4)$ is missed due to the
window bound.
\end{example}

\begin{example}[Circular window comparison]
Consider the trace
\[
\dots\;
e_1=[t_1:\wt(x)],\;
e_2=[t_1:\rel(\lk)],\;
e_3=[t_2:\acq(\lk)],\;
e_4=[t_2:\wt(x)].
\]
Assume a circular window of size $4$.
Suppose that when processing $e_1$ the framework stores it in the last cell,
i.e., at index $3$  and stores $e_2,e_3,e_4$ at the first three cells.
The synchronization edge $e_2\hb{\tr} e_3$ propagates to $t_2$ the current window index
of $t_1$'s latest event, namely $\CFT_{t_2}[t_1]=0$ (the index of $e_2$).
When processing $e_4$, the head is $h=2$.
Although the last write on $x$ is stored as $\epch^w_x=t_1@3$ and
$\CFT_{t_2}[t_1]=0$, the framework does \emph{not} report a race on $x$:
under the window-aware order $\cle_\window$ induced by $h$,
index $0$ is newer than index $3$ in the circular buffer, which means $3 <_\window 0$, so the race check remains sound.
\end{example}

\vspace*{-0.3in}
\input{algo-hb-window}

%% file: algo-ft.tex

{
\vspace*{0.8cm} 
\begin{algorithm}[h]
\small
\begin{multicols}{2}
\myfun{\init}{
         
	\ForEach{$t \in \threads{}$}{
		$\CFT_t \gets \bot$
	}
    
	\ForEach{$\lk \in \locks{}$}{
		$\CFT_\lk \gets \bot$
	}
	\ForEach{$x \in \vars{}$}{
		$\epch^{w}_x \gets null@0$ ;
		$\RFT_x \gets \bot$
	}
}

\myhandler{\rdhandler{$t$, $x$}}{
	$\CFT_t \gets \CFT_t[t \mapsto \CFT_t(t)+1]$ \;

	\lIf{$\epch^{w}_x \not\cle \CFT_t$}{\declare race} 
	$\RFT_x\gets \RFT_x[t \mapsto \CFT_t(t)]$  \;
    
}

\myhandler{\wthandler{$t$, $x$}}{
        $\CFT_t \gets \CFT_t[t \mapsto \CFT_t(t)+1]$ \;
	\lIf{$\RFT_x \not\cle \CFT_t$ or $\epch^{w}_x \not\cle \CFT_t$}{\declare race }
	$\epch^{w}_x\gets t@\CFT_t[t]$ \;
}

\myhandler{\acqhandler{$t$, $\lk$}}{
	$\CFT_t \gets \CFT_t[t \mapsto \CFT_t(t)+1]$\;
	$\CFT_t \gets \CFT_t \mx \CFT_\lk$  \;
}

\myhandler{\relhandler{$t$, $\lk$}}{
	$\CFT_t \gets \CFT_t[t \mapsto \CFT_t(t)+1]$ \;
	$\CFT_\lk \gets \CFT_t$  \; 
}

\vspace*{0.5cm}
\end{multicols}
\normalsize
\caption{\small Vector clock algorithm for detecting \acrhb-races}
\algolabel{djitp}
\end{algorithm}
}

%% file: algo-hb-window.tex
{
\vspace*{0.8cm} 
\begin{algorithm}[t!]

\small
\begin{multicols}{2}
\raggedcolumns
\setlength{\columnsep}{0.7em}
\setlength{\multicolsep}{0pt}

\myfun{\init}{
\tcp*[l]{Convention: all per-location/per-lock maps (e.g., $\RFT_x$, $\epch_x^w$, $\CFT_\lk$) default to $\text{null}$ if unset.}

$ W \gets [\textnormal{null}]^{\winsz}$\;
    $h\gets -1$\;
    \ForEach{$t \in \threads{}$}{
        $\CFT_t \gets \bot$\;
    }
}

\myfun{\textnormal{run($\tr$)}}{
    \ForEach{incoming event $e$ in $\tr$}{
        \textnormal{activate(e)}\;
    }
}

\myfun{\textnormal{activate(e)}}{
    $h\gets (h+1)\mod \winsz$\;
    $e_{\text{out}}\gets W[h]$\;
    \If{$e_{\text{out}}\neq \text{null}$}{
       \textnormal{collect(}$e_{\text{out}}$\textnormal{)}\;
    }
    Invoke \textbf{handler} on $e$\;
    $W[h]\gets e$\;
}

\myfun{\textnormal{collect(}$e_{\text{out}}$\textnormal{)}}{

    \tcp*[l]{Dispatch to the corresponding post-handler of the evicted event}
    \If{$e_{\text{out}}$ is a $\textsf{read}(t,x)$ event}{
        Invoke \textbf{post-handler} \rdhandler{$t,x$}\;
    }
    \If{$e_{\text{out}}$ is a $\textsf{write}(t,x)$ event}{
        Invoke \textbf{post-handler} \wthandler{$t,x$}\;
    }
    \If{$e_{\text{out}}$ is a $\textsf{release}(t,\lk)$ event}{
        Invoke \textbf{post-handler} \relhandler{$t,\lk$}\;
    }
}


\myhandler{\rdhandler{$t$, $x$}}{
    \lIf{$\epch^{w}_x \not\cle_{\window} \CFT_t$}{\declare race }
    \lIf{$\RFT_x=\text{null}$}{
        $\RFT_x\gets \bot$
    }
    $\RFT_x\gets \RFT_x[t \mapsto h]$
}

\myhandler{\wthandler{$t$, $x$}}{
    \lIf{$\RFT_x \not\cle_{\window} \CFT_t$ or $\epch^{w}_x \not\cle_{\window} \CFT_t$}{\declare race }
    $\epch^{w}_x\gets t@h$
}
\columnbreak

\myhandler{\acqhandler{$t$, $\lk$}}{
    \ForEach{$t^*\in \threads{\tr}$}{
        \If{$\CFT_\lk[t^*] >_{\window} \CFT_t[t^*]$}{
            $S_{\CFT_\lk[t^*]}.\textbf{add}(t)$\;
            $S_{\CFT_t[t^*]}.\textbf{remove}(t)$\;
            $\CFT_t[t^*]\gets \CFT_\lk[t^*]$\;
        }
        $S_{\CFT_\lk[t^*]}.\textbf{remove}(\lk)$\;
    }
    $\CFT_\lk\gets \text{null}$\;
}

\myhandler{\relhandler{$t$, $\lk$}}{
    $S_{\CFT_t[t]}.\textbf{remove}(t)$\;
    $S_{h}\gets \emptyset$\;
    $S_{h}.\textbf{add}(t)$\;
    $\CFT_t \gets \CFT_t[t \mapsto h]$\;
    $\CFT_\lk\gets \bot$\;
    \ForEach{$t^*\in \threads{\tr}$}{
        $S_{\CFT_t[t^*]}.\textbf{add}(\lk)$\;
        $\CFT_\lk[t^*]\gets \CFT_t[t^*]$\;
    }
}


\pohandler{\rdhandler{$t$, $x$}}{
    \lIf{$\RFT_x[t]=h$}{
        $\RFT_x[t]\gets -1$
    }
    \lIf{after eviction there is no read access on $x$ in W }{
        $\RFT_x[t]\gets \text{null}$\;
    }
}

\pohandler{\wthandler{$t$, $x$}}{
    \lIf{$\epch^{w}_x = t@h$}{
        $\epch^{w}_x\gets \text{null}$
    }
}

\pohandler{\relhandler{$t$, $\lk$}}{
    \ForEach{$t^*\in S_h$}{
        $\CFT_{t^*}[t]\gets -1$
    }
    \ForEach{$\lk^*\in S_h$}{
        $\CFT_{\lk^*}[t]\gets -1$
    }
    \lIf{$\CFT_\lk[t]=h$}{
        \ForEach{$t^*\in \threads{\tr}$}{
            $S_{\CFT_\lk[t^*]}.\textbf{remove}(\lk)$
        }
        $\CFT_\lk\gets \text{null}$\;
    }
    $S_h \gets \text{null}$\;
}

\vspace*{0.5cm}
\end{multicols}
\normalsize
\caption{\small Detecting short happens-before-races of length $\winsz$}
\algolabel{algo2}
\end{algorithm}
}

%% file: syncPhb.tex
\section{Sync-preserving Races}
\seclabel{sync-preserve}

Sync-preserving races are a general class of data races. As noted earlier, they subsume other race notions studied in the literature, including happens-before (HB) races (see \propref{hb-syncp}) and schedulable-happens-before (SHB) races~\cite{shb2018}: any trace that contains an HB (or SHB) race also contains a sync-preserving race, but the converse does not hold. Unlike HB, which reasons only about program order and synchronization order, sync-preserving races gain predictive power by additionally reasoning about feasible alternative executions induced by read-from relationships.

This additional predictive power comes at a computational cost. Although sync-preserving races admit streaming linear-time algorithms (e.g., \syncp~\cite{SyncP2021} and automata-based approaches~\cite{ang2024predictivemonitoringstrongtrace}), any such algorithm must, in the worst case, use \emph{linear space} (Theorem~3.2 of~\cite{SyncP2021}). Intuitively, the space blow-up is inherent because exploring feasible alternative executions requires maintaining causality information that may involve dependency chains of unbounded length.
In particular, this implies that sync-preserving races are \emph{non-local}, in sharp contrast to HB races.

\begin{proposition}
\proplabel{syncPncp}
There exists a well-formed subtrace \(\tr\) that begins with \(e_1\) and ends with \(e_2\), and well-formed traces
\(\tr_1\) and \(\tr_2\) such that \(\tr\) occurs as a subtrace of both \(\tr_1\) and \(\tr_2\), but \((e_1,e_2)\) is a
sync-preserving race in \(\tr_2\) and not in \(\tr_1\).
\end{proposition}

\begin{proof}
Consider the subtrace
\[
\tr \;=\;
[t_1\!:\; e_1=\wt(x)]
[t_3\!:\; \rd(x)]
[t_3\!:\; \rel(\lk)]
[t_2\!:\; \acq(\lk)]
[t_2\!:\; \rel(\lk)]
[t_2\!:\; e_2=\wt(x)].
\]
Now consider the well-formed traces
\[
\tr_1 \;=\; [t_3\!:\; \acq(\lk)]\,[t_1\!:\; \wt(y)]\,[t_2\!:\; \rd(y)]\;\cdot\;\tr,
\qquad
\tr_2 \;=\; [t_3\!:\; \acq(\lk)]\;\cdot\;\tr.
\]
We have \(\SPIdeal{\tr_1}(e_1,e_2)\cap \{e_1,e_2\}=\{e_1\}\), hence \((e_1,e_2)\) is \emph{not} a sync-preserving race in
\(\tr_1\).
In contrast, \(\SPIdeal{\tr_2}(e_1,e_2)\cap \{e_1,e_2\}=\emptyset\), which implies that \((e_1,e_2)\) \emph{is} a
sync-preserving race in \(\tr_2\).
\end{proof}
Therefore, a sliding-window algorithm by itself is insufficient for detecting short sync-preserving races:
even if two events co-occur within a window, whether they constitute a sync-preserving race in the full execution
cannot, in general, be determined from the window contents alone.
In this section, we show that sync-preserving races nevertheless admit an efficient short-race detector:
we combine a sliding window with compact, parameter-bounded summaries that retain exactly the out-of-window
information needed to evaluate window-local candidates.
We further demonstrate that the state-of-the-art streaming detector \syncp can be incorporated cleanly into our
framework, yielding a sound and complete algorithm for Short-\syncp{} races whose space usage depends only on the
fixed parameters (and not on \(\trsz\)).
 We begin by presenting a sketch of the \syncp algorithm.

\subsection{The \syncp Algorithm: a state-centric view}
\seclabel{syncp}
Among existing techniques for detecting sync-preserving races, \syncp~\cite{SyncP2021} is the most scalable:
it supports streaming detection and processes traces in linear time.
Automata-based approaches~\cite{ang2024predictivemonitoringstrongtrace}, while elegant and expressive, rely on
state-space constructions whose size is exponential in the relevant parameters, and thus do not scale to large
executions in practice.
Accordingly, we take \syncp as the underlying detector \(A_{\syncp}\) to be incorporated into our framework.

We give a self-contained overview of the \syncp{} algorithm at the level needed for our framework instantiation.
A complete description and all proofs appear in~\cite{SyncP2021}.
Recall from \secref{races} that a conflicting pair \((e_1,e_2)\) is a sync-preserving race iff
\(\SPIdeal{\tr}(e_1,e_2)\cap \{e_1,e_2\}=\emptyset\).
Thus, \syncp{} must (implicitly) compute \(\SPIdeal{\tr}(e_1,e_2)\) for conflicting pairs as the trace is streamed.
The key structural property is that the sets manipulated by \syncp{} (both \(\TLClosure{\tr}(\cdot)\) and \(\SPIdeal{\tr}(\cdot,\cdot)\))
are downward-closed w.r.t.\ the trace order \(\tho{\tr}\), and hence can be represented compactly by vector timestamps.

\myparagraph{Online maintenance of \(\TLClosure{\tr}(e)\).}
As events arrive, \syncp{} maintains a vector timestamp for the \(\TLClosure{\tr}(e)\) of each processed event \(e\).
This is done incrementally using a small set of vector clocks, in particular:
for each memory location \(x\), a clock \(\mathbb{LW}_x\) that stores the vector timestamp representing
\(\TLClosure{\tr}(\ell)\), where \(\ell\) is the most recent write to \(x\) observed so far.

Intuitively, for a read event \(e=\rd(x)\), feasibility requires \(e\) to ``see'' some write to \(x\), and therefore
\(\TLClosure{\tr}(e)\) must include the most recent write summary, i.e., its timestamp is pointwise \(\sqsupseteq \mathbb{LW}_x\).
(Updates for program order and synchronization are handled similarly by joining the appropriate vector timestamps; see~\cite{SyncP2021}.)

\myparagraph{Race detection with $\SPIdeal{\tr}(e_1,e_2)$ predicate.}
To decide whether \((e_1,e_2)\) is a race, \syncp{} must compute \(\SPIdeal{\tr}(e_1,e_2)\), which is defined via a fixed-point
construction (see \secref{races}). Operationally, this fixed point is driven by repeated discovery of same-thread acquire pairs and
insertion of the matching releases. Accordingly, \syncp{} stores two classes of event summaries, both in trace order:
\begin{itemize}[leftmargin=1.5em,itemsep=0.2em,topsep=0.2em]
\item \emph{Critical-section summaries.} For each thread \(t\) and lock \(\lk\), a list \(\cs_{t,\lk}\) whose entries have the form
\((\TLClosure{\tr}(a),\,\TLClosure{\tr}(\match{\tr}(a)))\), where \(a\) ranges over acquire events on \(\lk\) by \(t\).
\item \emph{Access summaries.} For each thread \(t\), access type \(a\in\{\rd,\wt\}\), and location \(x\),
a list \(\access_{t,a,x}\) containing entries \((\TLClosure{\tr}(\prev{\tr}(e)),\,\TLClosure{\tr}(e))\) for each access event \(e\) of type \(a\) on \(x\) by \(t\).
These lists allow \syncp{} to check a newly observed access against relevant prior conflicting accesses.
\end{itemize}

\myparagraph{Incrementality and monotonicity.}
A central reason \syncp{} admits streaming, linear-time processing is that the fixed-point computations for different conflicting pairs
can be shared. \syncp{} exploits monotonicity by maintaining (i) cached vector timestamps that summarize the most recent ideal computed
for each conflict pattern, and (ii) forward-only pointers into the \(\cs\)- and \(\access\)-lists so that each entry is processed only
a bounded number of times.
(We omit the full bookkeeping details here; see~\cite{SyncP2021}.)

\myparagraph{Complexity.}
Overall, \syncp processes a trace \(\tr\) in time
\(O(\trsz\numthr^2 + \numacq\numvar\numthr^3)\),
and uses space
\(O(\numvar\numlk\numthr^3 + \trsz\numthr\log \trsz)\)%
\footnote{The space bound for \syncp reported in~\cite{SyncP2021} contains a small error (as confirmed by the authors).
The bound stated here is the corrected one.}.
In particular, the \(\trsz\numthr\log\trsz\) term reflects that \syncp maintains event-associated vector timestamps whose
components store global local times (hence \(\Theta(\log\trsz)\) bits per component) for a linear number of events.
This inherent linear-space dependence motivates our short-race framework, which seeks to make the maintained state depend
on the short-race parameter \(\winsz\) rather than on the full trace length \(\trsz\).

\subsection{Detecting Short \syncp Races}
\seclabel{syncp-short}

The linear space lowerbound for any streaming algorithm to detect sync-preserving races~\cite{SyncP2021} suggests that efficient algorithms for detecting short sync-preserving races using sub-linear space could have significant benefits. As in the HB case, our goal in instantiating the framework is to preserve---and ideally even improve---the running
time of the underlying \syncp detector, while eliminating its space dependence on the trace length \(\trsz\).
However, as Prop.~\ref{prop:syncPncp} indicates, sync-preserving races are non-local, and deciding whether a
window-local conflicting pair forms a race may require information from outside the window.
Accordingly, we do not aim to eliminate \emph{all} dependence on \(\trsz\); rather, our primary objective is to remove
\emph{linear} dependence on \(\trsz\), allowing at most sublinear auxiliary space beyond the window and the fixed
parameters.

We first observe that sync-preserving races are well-defined on well-formed traces:
in particular, there are no unmatched releases, so the \(\SPIdeal{\cdot}(\cdot,\cdot)\)-based
race predicate is unambiguous.

\begin{proposition}
\proplabel{syncPcp}
Let \(\tr\) be a well-formed trace and let \((e_1,e_2)\) be a conflicting pair in \(\tr\).
For any well-formed trace \(\tr'\) such that \(\tr\) is a subtrace of \(\tr'\), \((e_1,e_2)\) is a sync-preserving
race in \(\tr'\) if and only if \((e_1,e_2)\) is a sync-preserving race in \(\tr\).
\end{proposition}

Thus, a straightforward baseline is to maintain an unbounded window that, for every \emph{unmatched} acquire, retains
all subsequent events until that acquire is matched by its corresponding release (so that the suffix since the oldest
unmatched acquire is kept).
However, this offers no space savings: if an acquire remains unmatched for a long time, the retained suffix can grow
linearly with the trace, resulting in essentially the same space usage as vanilla \syncp.

Nevertheless, this baseline provides a useful guiding principle for our design:
we aim to maintain, for the events currently in the window, exactly the information that \syncp would maintain about
those events when run on the corresponding well-formed suffix of the execution.
In other words, our state is designed so that, when restricted to window-resident events, it is semantically
equivalent to \syncp's state on the same events, enabling us to judge window-local candidate pairs exactly as \syncp
would in the full trace.
The high-level intuition is that the causality chains induced by read-from edges---which are needed for sync-preserving
race reasoning---can be maintained online in a \emph{compressed} form as vector timestamps.

As a result, we do not need to retain all memory-access events in the unbounded suffix that the baseline would keep
after an unmatched acquire; instead, their effect is summarized by  vector timestamps.

This intuition is reflected in the design of \syncp:
\syncp detects races by evaluating the \(\SPIdeal{\cdot}(\cdot,\cdot)\)-based predicate via vector clock updates: 

As discussed above, it maintains, for each event, a vector-timestamp representation of its closure \(\TLClosure{\tr}(\cdot)\),
and uses these closures to drive the fixed-point computation underlying \(\SPIdeal{\tr}(\cdot,\cdot)\) during race checking.
The \(\TLClosure{\tr}(\cdot)\) closures are maintained in a manner analogous to the thread vector clocks in {\fasttrack}:
\syncp keeps only \(O(1)\) vector clocks per fixed parameter combination (e.g., per thread/location), and each
vector-clock component stores a global local time, hence requiring \(\Theta(\log \trsz)\) bits. This implies that maintaining \(\TLClosure{\tr}(\cdot)\) is not the primary source of inefficiency---in particular,
it is not what drives \syncp's linear-space usage---and we can therefore incorporate this component into our
framework without modification. Moreover, retaining \(\TLClosure{\tr}(\cdot)\) provides exactly the causal summary(a downward-closed causal closure represented as a vector timestamp) we need to account for dependency
chains that extend beyond the current window. 

Consequently, the only out-of-window events that must be stored explicitly are those whose information cannot be
compressed away, namely acquire events whose matching releases have not yet been observed. Intuitively, retaining these unmatched acquires is necessary to ``complete'' the window: it allows us to interpret
the window contents as if they formed a well-formed trace, so that race checks performed within the window have the
same semantics as in the full execution. 
We have now addressed how to \emph{handle} non-locality.
To obtain a genuine improvement, we must still design the framework operators so that the resulting monitor is
(i) \emph{correct} (sound and complete for Short-\syncp{} races), and (ii) \emph{efficient}, removing the linear
\(\trsz\)-space dependence without increasing the asymptotic running time of \syncp.
To this end, we first identify which parts of \syncp's state are event-indexed and must be allocated/reclaimed at the
window boundary.
\input{algo-syncp-window}

\myparagraph{Event-associated view (used by our framework).}
For our purposes, it is convenient to view the \syncp{} state as a collection of \emph{records dynamically associated with events}.
This perspective makes clear which parts of the state ``enter'' with an event and which parts remain referenced by later computation,
which will be crucial for defining boundary operations in our short-race framework.

Let \(e\) be a processed event.
\begin{enumerate}[leftmargin=1.7em,itemsep=0.25em,topsep=0.25em]
\item If \(e=\rd(t,x)\) (resp.\ \(e=\wt(t,x)\)) by thread \(t\),  then \(e\) contributes one access-summary entry
\((\TLClosure{\tr}(\prev{\tr}(e)),\,\TLClosure{\tr}(e))\) appended to \(\access_{t,\rd,x}\) (resp.\ \(\access_{t,\wt,x}\)).
Moreover, if \(e\) is the latest access of its kind by \(t\) on \(x\), then future conflict processing may reference the \emph{list head}
and the corresponding cached ideals for conflicts against other threads.
\item If \(e=\acq(t,\lk)\) then \(e\) contributes a critical-section entry in \(\cs_{t,\lk}\) whose second component is
temporarily \(\textsf{null}\) until the matching release arrives.
\item If \(e=\rel(t,\lk)\), then \syncp{} completes the corresponding critical-section entry by filling in
\(\TLClosure{\tr}(\match{\tr}(e))\), and subsequent ideal computations may reference this completed entry via forward-only pointers.
\end{enumerate}
In addition, \syncp also maintains information associated with fixed parameters rather than by individual events.
Specially, it keeps a vector clock for each thread $t$, each memory location $x$, and each lock $l$.
Unlike event-local records, these per-thread/per-location/per-lock structures are persistent in our framework as they are necessary to conclude short race witness so they are not created or reclaimed by our window-boundary operations.

In summary, \syncp{} can be understood as maintaining a set of event-associated records (access entries and critical-section entries),
together with a small amount of monotone bookkeeping (cached ideals and forward pointers) that incrementally consumes these records.

\algoref{algo3} gives a high-level view of our framework. The framework first initialize a size-\(\winsz\) window as well as the data structures of \syncp that are reserved for the fixed-parameters.

The wrapper functions \allocate{} and \deallocate{} allocate/free event-local records as discussed above.
To make \deallocate{} safe, our implementation of \collect{} uses two internal bookkeeping helpers, \spill{} and \replay{}.
Intuitively, \spill{} extracts the information about unmatched acquires that would otherwise be reachable only
via the event-local list $\cs_{t,\lk}$ upon eviction, and stores it in a separate variable $\Open_\lk$; This enables \deallocate{} to detach and free the pointers from $\cs_{t,\lk}$ without losing the
unmatched-acquire information needed later, so that the standard \syncp{} logic can proceed unchanged.
Later, when a matching release enters the window, \replay{} re-injects the spilled information from $\Open_\lk$ back into
$\cs_{t,\lk}$ so that the standard \syncp{} handler can proceed unchanged and remain sound.

For each incoming event \(e_{\text{new}}\), the \activate{} operator advances sliding window, inserts \(e_{\text{new}}\), identifies the evicted event \(e_{\text{old}}\), and invokes
$\collect{}(e_{\text{old}})$ at the window boundary.

The monitor then invokes the standard \syncp{} handler on \(e_{\text{new}}\) using only window-resident records as well as fixed-parameter records.

\algoref{algo3} achieves the same runtime as \syncp (and even better when $\winsz$ is sufficiently small) while eliminating the linear space dependence on the trace length required by \syncp. The proof is presented in the Appendix.
\begin{theorem}
    \thmlabel{algo3-correctness}
    \algoref{algo3} correctly solves the short sync-preserving-race detection problem. It runs in time $O(\trsz\numthr(\min(\winsz,\numthr))+\numacq\numthr(\min(\winsz,\numvar\numthr^2)))$ and uses space $O(\min(\winsz,\numvar\numthr)\numlk\numthr^2+(\winsz+\numvar+\numlk)\numthr\log\trsz)$. 
\end{theorem}

%% file: algo-syncp-window.tex
{
\vspace*{0.8cm} 
\begin{algorithm}[htbp]
\small
\begin{multicols}{2}

\myfun{\init}{
\tcp*[l]{Convention: all per-location/per-lock/per-tuple maps default to $\text{null}$ if unset.}
$W \gets [\textnormal{null}]^{\winsz}$\;
$h \gets 0$\;

\tcp*[l]{Spilled open-acquire record per lock: stores the unique open entry \(\TLClosure{\tr}(\acq,\textsf{null})\) if any.}
\ForEach{$\lk \in \locks{}$}{
    $\Open_\lk \gets \text{null}$\;
}

\tcp*[l]{Initialize the underlying \syncp states for every fixed-parameter }

}

\myfun{\textnormal{run($\tr$)}}{
    \ForEach{incoming event $e$ in $\tr$}{
        \textnormal{activate(e)}\;
    }
}


\myfun{\textnormal{activate(e)}}{
    $h \gets (h+1)\mod \winsz$\;
    $e_{\text{out}} \gets W[h]$\;

    \If{$e_{\text{out}}\neq \text{null}$}{
        \textnormal{collect(}$e_{\text{out}}$\textnormal{)}\;
    }

    \tcp*[l]{Replay: if a release enters but its matching acquire is outside the window, complete the window.}
    \lIf{$e$ is a $\rel(t,\lk)$ event $\land\Open_\lk\neq \text{null}$}{
        \textnormal{replay(}$e$\textnormal{)}\;
    }

    \tcp*[l]{Allocate event-local records needed by \syncp for this event.}
        \textnormal{allocate(}$e$\textnormal{)}\;

    \tcp*[l]{Underlying detector step: invoke the standard \syncp handler (unchanged).}
    Invoke \textbf{\syncp.handler} on $e$\;

    $W[h] \gets e$\;
}
\columnbreak

\myfun{\textnormal{spill(}$e_{\text{out}}$\textnormal{)}}{
    \tcp*[l]{spill only when an acquire is evicted while still unmatched.}
    \lIf{$e_{\text{out}}$ is a $\textsf{acquire}(t,\lk)$ event $\land\ \match{\tr}(e_{\text{out}})\notin W$}{
    $\Open_\lk \gets \TLClosure{\tr}(e_{\text{out}}, \textsf{null}).$
        \tcp*[f]{spill the open entry \(\TLClosure{\tr}(e_{\text{out}}, \textsf{null})\) from $\cs_{t,\lk}$ into $\Open_\lk$}
    }
}

\myfun{\textnormal{replay(e)}}{
    \tcp*[l]{Reinsert the spilled open CS-entry from $\Open_\lk$ back into $\cs_{t,\lk}$.}
\tcp*[l]{At the moment a release is processed with its matching acquire out of window, $\cs_{t,\lk}$ has no window-resident open entry.}
    $\cs_{t,\lk}\gets\Open_\lk$\;
    $\Open_\lk\gets\text{null}$ 
    
}

\myfun{\textnormal{collect(}$e_{\text{out}}$\textnormal{)}}{
    \tcp*[l]{(1) Spill unmatched acquires if needed (to preserve well-formedness of future windows).}
    \textnormal{spill(}$e_{\text{out}}$\textnormal{)}\;

    \tcp*[l]{(2) Reclaim event-local records and detach any references from monotone bookkeeping.}
    \textnormal{deallocate(}$e_{\text{out}}$\textnormal{)}\;
}


\myfun{\textnormal{allocate($e$)}}{
\tcp*[l]{Allocate only the event-local records that \syncp would create for $e$ (nodes in $\access$ or $\cs$).}
\tcp*[l]{All persistent components (e.g., $\CFT_t$, $\mathbb{LW}_x$, cached ideals, forward-only pointers) are updated by \syncp.handler.}
}

\myfun{\textnormal{deallocate(}$e_{\text{out}}$\textnormal{)}}{
\tcp*[l]{Free all event-local records associated with $e_{\text{out}}$ that remain in-window (access/CS nodes).}
\tcp*[l]{Also detach any references from cached ideals/pointers so that no dangling pointers remain.}
}

\vspace*{0.8cm}
\end{multicols}
\normalsize
\caption{\small Detecting short sync-preserving races of length $\winsz$ (framework instantiation with \syncp).}
\algolabel{algo3}
\end{algorithm}
}

%% file: evaluation.tex
\section{Experiments}
\seclabel{experiments}

\myparagraph{Overview.}
In this section, we present our experimental evaluation of the two short race algorithms on a standard benchmark suite for race prediction. Since \fasttrack is already highly efficient, our study of short \fasttrack primarily investigates whether its theoretical improvement in space complexity translates into any practical performance gains. By contrast, the \syncp algorithm—although the state-of-the-art solution for predicting sync-preserving races—is known to suffer from intractability even on moderate or small benchmarks (with trace lengths under 1M). We therefore expect the short \syncp algorithm to substantially improve the practicality of predicting sync-preserving races. Consequently, the \textbf{bulk} of this section focuses on analyzing and understanding the behavior of short \syncp. We start by summarizing the results obtained from running the short \fasttrack algorithm, followed by the results for short \syncp algorithm.

\myparagraph{Setup.}
The experiments were conducted on a MacBook Pro equipped with 36GB of RAM and an M3 chip. To ensure reproducibility and consistency across different environments, we utilized a Docker container to encapsulate a Linux-based environment.
We implemented the two algorithms proposed in the paper on top of \rapid in JAVA. \rapid is an offline dynamic analysis framework designed to analyze execution trace logs. It includes a built-in optimized implementation of \fasttrack and provides the only publicly available implementation of \syncp.

\myparagraph{Benchmarks.}
Our experiments use execution traces from~\cite{Shi2024}, comprising \textbf{153 programs in total}:
(i) \textbf{30 Java programs} drawn from multiple benchmark suites, including IBM Contest~\cite{Farchi2003},
DaCapo~\cite{DaCapo2006}, SIR~\cite{SIR2005}, Java Grande Forum~\cite{sen2005detecting}, and several standalone benchmarks;
and (ii) \textbf{123 OpenMP/HPC benchmarks} drawn from OmpSCR~\cite{dorta2005openmp}, DataRaceBench~\cite{liao2017dataracebench},
DataRaceOnAccelerator~\cite{schmitz2019dataraceonaccelerator}, the NAS Parallel Benchmarks~\cite{nasbenchmark}, CORAL~\cite{coral},
ECP proxy applications~\cite{ecp}, and the Mantevo suite.
All traces record accesses to shared variables and synchronizations via acquire/release events on lock objects only.

\myparagraph{Trace characteristics.}
Our benchmark suite contains a total of 153 traces. To better illustrate the results, we filter out 59 traces with fewer than 10M events. Among the remaining 92 traces, the largest contains 900M events, while the average and median trace sizes are 123M and 70M events, respectively. For clarity of presentation, we group the traces into categories: OMPRacer (12), DRACC (8), DataRaceBench (30), HPCBench (34), CoMD (6), and Other (4).

\myparagraph{Configurations and selection of short race lengths.} 
We implemented \algoref{algo2} and \algoref{algo3} on \rapid's \fasttrack and \syncp implementation, referred to as \shortfasttrack and \shortsyncp, respectively. Each tool was run with a 28GB memory limit and a 1-hour time limit.

We use \rapid's \fasttrack implementation as the baseline for \shortfasttrack. However, \rapid's \syncp implementation is highly inefficient, timing out on most benchmarks in this study even under significantly higher resource limits (4 hours and 400GB~\cite{Shi2024}). In contrast, our \shortsyncp implementation includes substantial optimizations, allowing it to complete 37 benchmarks contains more than 10 million events (as well as filtered benchmarks which are under 10 million events) with window size set to 100\% of the trace length, even within our more constrained resource limits. Therefore, we use \shortsyncp with a race length of 100\% of the trace length as the baseline for evaluating \shortsyncp. For the rest of the section, we will refer to this baseline as \syncp. 
For \shortfasttrack, we ran each benchmark trace under two short-race bounds—$3000$ and $30000$ ($3$K, $30$K)—repeating each configuration three times. We chose these race lengths for two reasons. First, prior work \cite{RPT2023}, which samples subtraces to detect races, suggests that a bound of $30$K is broadly effective for detecting happens-before races but does not evaluate a smaller bound such as $3$K; we therefore include $3$K to assess its effectiveness. Second, these bounds allow us to demonstrate potential space savings by storing vector-clock entries in the Java primitive \texttt{short} rather than \texttt{int}.

In contrast, for \shortsyncp, we ran each trace with larger race lengths—$1$ million and $10$ million($1$M and $10$M) —each configuration repeated three times. We chose these trace lengths also for two reasons: (i) prior work~\cite{SyncP2021} reports that sync-preserving races can span long distances; and (ii) after our optimizations, \syncp completes all benchmarks up to $10$M events without out-of-memory (OOM) errors, indicating that running \shortsyncp at this length is feasible. This is also how we recommend developers to select the race length.

\myparagraph{Soundness and completeness.}
As discussed in \secref{frame}, our algorithms and the implementations satisfy the following guarantees.
\emph{Soundness}: every race they report is a genuine HB-race or sync-preserving race.
\emph{Completeness}: for any span bound $w$, they detect every short race whose span is at most $w$.

\myparagraph{\shortfasttrack.}
\input{hbsl_10k30k_group_config_rows}
\tabref{hbtab} presents the results of running \shortfasttrack on the benchmark suites. All configurations successfully complete every benchmark trace within the time limit. We gauge the efficiency of our configurations with the 
running time, number of racy events and location reported, and the maxium heap usage. For every trace, we average the result from the three repeated runs and then aggregate the number by groups. 
Although \algoref{algo2} has the same asymptotic running time as \fasttrack, we expect \shortfasttrack to be slower due to the additional overhead of managing event removal from the window. However, as shown in \tabref{hbtab}, this is not always the case. In the DRACC benchmarks, \shortfasttrack does not show a noticeable improvement over \fasttrack. In contrast, in all the other benchmark suites, the smaller trace length $3$K,  achieves clear gains in running time. However, we also observe that such gains always comes with reduction in race events and race variables reported. 

Our experiments with \shortfasttrack do not provide strong evidence for reducing memory usage across benchmarks. In several cases, the larger trace length configuration ($30$K) uses more heap than \fasttrack. One factor for this scenario is that \fasttrack’s memory space grows linearly with fixed parameters (number of threads, locks, and variables) and only logarithmically with execution length --- more precisely, with the per-thread maximum number of lock operations. Consequently, windowing offers little savings unless those fixed parameters are large or the trace induces heavy per-thread lock operations. In addition, the bookkeeping needed to maintain the window can outweigh its benefits; our current implementation is not yet aggressively optimized for space, which can further amplify this effect. We remark that the memory savings would be more pronounced if we compared against a \fasttrack implementation that stores vector-clock entries as \texttt{long} rather than \texttt{int}. However, in our experiments the largest trace has $900$M events—well within the range of a 32-bit \texttt{int}, so the contrast is less apparent. 

\myparagraph{\shortsyncp.}
{\rapid}’s \syncp implementation is the only publicly available implementation of \syncp and is highly inefficient. As reported by Shi et al.~\cite{Shi2024}, it failed to coverage sufficient event to explore race conditions and even failed to complete the initialization of its data structures on many traces with fewer than one million events. Our implementation of \algoref{algo3} (\shortsyncp) incorporates several optimizations and, with a window size of $100\%$, successfully analyzes a substantially larger fraction of traces than \rapid’s \syncp. However, it still failed to finish the analysis on majority of the traces.
\begin{figure}[t]
  \centering
  \begin{subfigure}[t]{0.49\columnwidth}
    \centering
    \includegraphics[width=\linewidth]{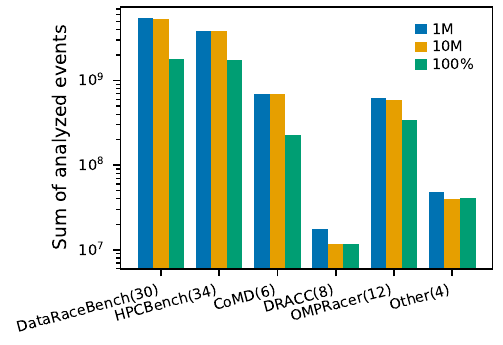}
    \subcaption{Total analyzed events by group (log scale).}
  \end{subfigure}\hfill
  \begin{subfigure}[t]{0.49\columnwidth}
    \centering
    \includegraphics[width=\linewidth]{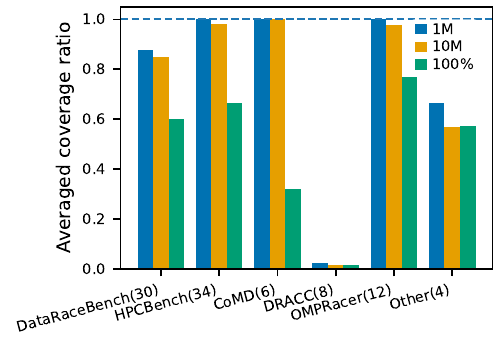}
    \subcaption{Average analyzed/baseline by group (linear).}
  \end{subfigure}
  \caption{Events count across benchmark groups. (a): Total events analyzed(log). (b): Average of normalized analyzed events relative to each trace’s size across each benchmark group.}
  \figlabel{coverage}
\end{figure}

\myparagraph{Trace coverage.}
Therefore we first quantify the improvement in trace coverage (number of events analyzed) achieved by \shortsyncp. In \figref{coverage}, the left panel reports the number of events analyzed per benchmark group, while the right panel shows the mean coverage ratio, defined as the average over traces of (analyzed events ÷ trace length). \shortsyncp with window sizes of 1M and 10M analyzes substantially more events than \syncp and attains a markedly higher coverage on most benchmarks. The two window sizes yield similar coverage because, under our resource limits, both often complete the full analysis on many benchmarks. For the DRACC group, however, windowing offers little improvement in coverage; our investigation suggests this is due to the inherent complexity of \syncp on traces dominated by synchronization events. Overall, we conclude that \shortsyncp significantly improves the coverage. 

\myparagraph{Running time.}
\begin{figure}[t]
  \centering
  \begin{subfigure}[t]{0.49\columnwidth}
    \centering
    \includegraphics[width=\linewidth]{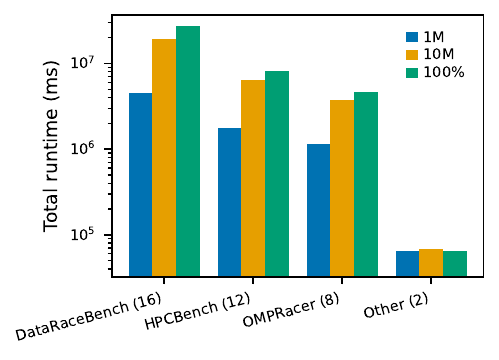}
    \subcaption{Total run time by group (ms).}
  \end{subfigure}\hfill
  \begin{subfigure}[t]{0.49\columnwidth}
    \centering
    \includegraphics[width=\linewidth]{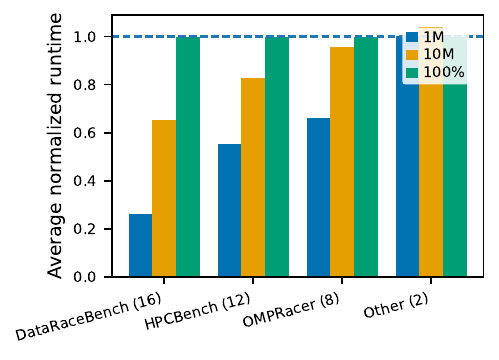}
    \subcaption{Average normalized runtime group.}
    \label{fig:events-norm}
  \end{subfigure}
  \caption{(a): Total running time (log). (b): Average normalized running time relative to the baseline.}
\figlabel{runtime}

\end{figure}
Our analysis suggests that \shortsyncp improves runtime complexity; we now assess its practical impact. \figref{runtime} reports total runtime and averaged normalized runtime per benchmark group (relative to \syncp). For fairness, we exclude (i) the 42 traces where \syncp ran out of memory, to avoid bias from early termination, and (ii) the 14 traces where all configurations timed out (including all eight DRACC traces). Unlike the coverage results, the 1M configuration incurs far less overhead than both 10M and \syncp. On DataRaceBench and HPCBench, its averaged normalized runtime is $26\%$ and $55\%$ of \syncp’s, respectively. The 10M configuration also outperforms \syncp but with higher overhead than 1M. 

\myparagraph{Race prediction power.}
We next examine whether higher coverage translates into higher race exposure. Because benchmarks vary widely in absolute race counts (some traces have up to millions), we report average normalized detection rates for two metrics in \tabref{racytab}: (1) racy pairs detected and (2) racy locations detected. For each trace, we record the counts reported by each configuration, normalize them by the baseline’s counts on the same trace, and then average these ratios across traces from the same group. Notably, the 10M configuration exhibits substantially higher detection power—up to 1000x on both metrics in HPCBench, and 
165x (racy pairs) and 7.2x (racy locations) in DataRaceBench, which comprises most of our traces. The 1M configuration also improves markedly over the baseline, though it generally trails 10M. Therefore, even 1M configuration in general achieves better coverage, the larger window yields greater exposure, indicating that additional temporal context is important for uncovering races.
\input{racy_norm_tables_combined}

\myparagraph{Memory Usage.}
\begin{figure}[t]
  \centering
  \begin{subfigure}[t]{0.49\columnwidth}
    \centering
    \includegraphics[width=\linewidth]{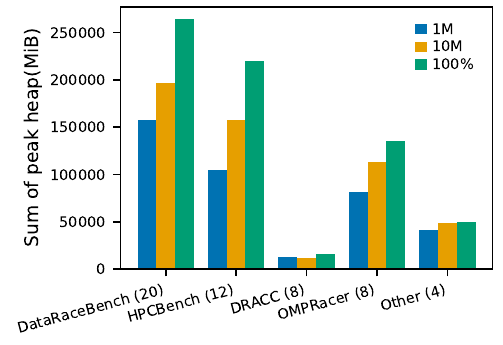}
    \subcaption{Sum of average max heap usage.}
  \end{subfigure}\hfill
  \begin{subfigure}[t]{0.49\columnwidth}
    \centering
    \includegraphics[width=\linewidth]{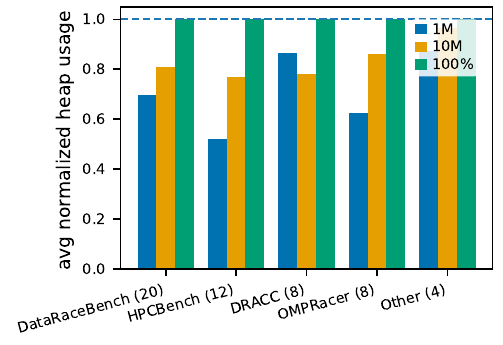}
    \subcaption{Average normalized max heap usage.}
  \end{subfigure}
  \caption{(a): Sum of average max heap usage. (b): Normalized average of analyzed events relative baseline's max heap usage.}
  \figlabel{memory}
\end{figure}
Finally, we evaluate memory usage. Across the 92 benchmarks, the \syncp baseline triggered out-of-memory (OOM) on 42 traces, whereas neither \shortsyncp configuration (1M, 10M) reported any OOMs. This already indicates a substantial memory-efficiency advantage of \shortsyncp over \syncp. For the remaining 50 traces, we quantify memory by logging GC behavior and recording the maximum heap usage for each run. As shown in \figref{memory}, \shortsyncp achieves large reductions in peak heap usage even on traces that are not memory-stressed.


%% file: hbsl_10k30k_group_config_rows.tex
\begin{longtable}{l l r r r r}
\toprule
Benchmark group & Config & Avg time (ms) & Racy events & Racy vars & Max heap (MiB) \\
\midrule
\endfirsthead
\toprule
Benchmark group & Config & Avg time (ms) & Racy events & Racy vars & Max heap (MiB) \\
\midrule
\endhead
\midrule
\multicolumn{6}{r}{\emph{Continued on next page}}\\
\endfoot
\bottomrule
\endlastfoot
\multirow{3}{*}{DataRaceBench} & 3K & 5,377,487 & 64,546,596 & 46,166 & 16,799.5 \\
& 30K & 6,152,631 & 172,952,750 & 51,569 & 24,198.3 \\
& 100\% & 6,169,788 & 172,990,978 & 55,077 & 17,552.2 \\
\addlinespace[2pt]
\multirow{3}{*}{HPCBench} & 3K & 2,844,152 & 4,741,389 & 43,049 & 33,199 \\
& 30K & 3,123,535 & 34,584,174 & 84,672 & 36,828 \\
& 100\% & 4,633,655 & 775,349,288 & 9,289,456 & 36,230 \\
\addlinespace[2pt]
\midrule
\multirow{3}{*}{CoMD} & 3K & 551,748 & 2,148,054 & 8,878 & 3,210.9 \\
& 30K & 618,002 & 15,941,852 & 17,855 & 3,808 \\
& 100\% & 857,893 & 603,504,406 & 67,114 & 3,426.5 \\
\addlinespace[2pt]
\midrule
\multirow{3}{*}{DRACC} & 3K & 629,986 & 401,225,627 & 2,172 & 4,793 \\
& 30K & 696,393 & 401,225,627 & 2,172 & 7,167.7 \\
& 100\% & 609,849 & 401,225,627 & 2,172 & 4,356.7 \\
\addlinespace[2pt]
\midrule
\multirow{3}{*}{OMPRacer} & 3K & 460,955 & 13,239,193 & 21,888 & 8,900.1 \\
& 30K & 517,229 & 64,256,697 & 261,779 & 9,526.5 \\
& 100\% & 761,538 & 332,176,279 & 1,280,965 & 9,821.8 \\
\addlinespace[2pt]
\midrule
\multirow{3}{*}{Other} & 3K & 114,734 & 56,748 & 3,443 & 2,111.9 \\
& 30K & 118,589 & 329,138 & 8,402 & 3,063 \\
& 100\% & 122,657 & 13,967,485 & 9,382 & 2,258.7 \\
\end{longtable}
\captionof{table}{Per–benchmark-group results for race length configurations (10K, 30K, 100\%); 100\% is the {\fasttrack} implementation.
Each metric (running timeRacy events, Racy vars, Max heap ) is the sum across benchmarks in that group, computed after averaging per benchmark across its runs. 
We include only traces whose number of events is at least 10M.}
\tablabel{hbtab}

%% file: racy_norm_tables_combined.tex
\begin{table}[t]
  \centering
  \begin{minipage}{.48\linewidth}
    \centering
    \small \textbf{(a) Averaged racy events ratio }\\[-0.2em]
    \input{racy_events_normalized_tabular.tex}
  \end{minipage}\hfill
  \begin{minipage}{.48\linewidth}
    \centering
    \small \textbf{(b)Averaged racy locations ratio}\\[-0.2em]
    \input{racy_locations_normalized_tabular.tex}
  \end{minipage}
  \caption{Average normalized racy events and racy locations per group (ratio to 100\%).}
  \tablabel{racytab}
\end{table}

%% file: racy_events_normalized_tabular.tex
\begin{tabular}{lcccr}
\toprule
 & 1M & 10M & 100\% & n \\
bench group &  &  &  &  \\
\midrule
DataRaceBench & 54.47x & 165.61x & 1.00x & 30 \\
HPCBench & 37.40x & 1279.43x & 1.00x & 34 \\
CoMD & 3.55x & 3.68x & 1.00x & 6 \\
DRACC & 1.32x & 1.00x & 1.00x & 8 \\
OMPRacer & 1.84x & 1.75x & 1.00x & 12 \\
Other & 1.17x & 0.98x & 1.00x & 4 \\
\bottomrule
\end{tabular}

%% file: racy_locations_normalized_tabular.tex
\begin{tabular}{lcccr}
\toprule
 & 1M & 10M & 100\% & n \\
bench group &  &  &  &  \\
\midrule
DataRaceBench & 7.23x & 7.21x & 1.00x & 30 \\
HPCBench & 36.93x & 1278.72x & 1.00x & 34 \\
CoMD & 1.00x & 1.00x & 1.00x & 6 \\
DRACC & 1.00x & 1.00x & 1.00x & 8 \\
OMPRacer & 0.98x & 1.03x & 1.00x & 12 \\
Other & 1.15x & 0.99x & 1.00x & 4 \\
\bottomrule
\end{tabular}

%% file: relatedwork.tex
\section{Related Work}
\seclabel{relate}
\myparagraph{Static and Dynamic Race detection}{
Data races constitute a prevalent—and notoriously hard-to-detect—class of concurrency errors. The existence of a race, which is commonly defined as a conflicting pair of events that are unordered under happens before, does not only indicates non-deterministic behavior, but also indicates more serious issues including data corruption \cite{boehmbenign2011} and unsound compiler optimization.  A substantial body of research has produced influential static and dynamic approaches for race detection. Static analyses \cite{Naik:2006:ESR:1133255.1134018,racerd2018,flanagan2000type,Abadi:2006:TSL:1119479.1119480,voung2007relay} aim to verify data-race freedom using type-based systems or interprocedural analyses over data and control flow. As verifying concurrent programs is computationally hard, practical static analyzers typically either report a large number of false postives or can miss intensive amount of races. Notably, RacerD~\cite{racerd2018} and its successor~\cite{racerdx2019} have demonstrated improved practical precision, substantially reducing false alarms.
Dynamic analyses operate on concrete executions, so detection is bounded by the behaviors exercised at run time. Beyond lock-based analyses \cite{savage1997eraser} that check lock discipline, dynamic analysis primarily asks whether a given execution contains a race. In this setting, happens-before (HB) remains the canonical notion, underpinning a large body of research \cite{elmas2007goldilocks,djit,fasttrack,Pozniansky:2003:EOD:966049.781529,Fidge:1991:LTD:112827.112860} and industrial tools such as \tsan.
}

\myparagraph{Predictive dynamic race detection}{
A key approach to enhancing dynamic analysis is predictive race detection: given one execution, infer whether some feasible alternative reordering of the same events would expose a race, rather than reasoning only about the observed order. Our work focuses on happens-before (HB) races \cite{lamport1978time} and sync-preserving races \cite{SyncP2021}, which reflect two complementary paradigms: partial-order–based exploration versus computing correct reorderings. HB has been refined to schedulable HB (SHB) \cite{shb2018}, which guarantees that each reported race is realizable in a correct reordering of the trace. Causal Precedence (CP) \cite{cp2012} and Weak Causal Precedence (WCP) \cite{wcp2017} further relax HB while retaining tractable algorithms—polynomial time for CP and linear time for WCP. In a similar spirit, M2 \cite{PavlogiannisPOPL20} and OSR \cite{Shi2024} define race classes discoverable via constrained correct reorderings, akin to \syncp.
}

\myparagraph{Windowing}{
Windowing is a natural technique to reduce the overhead of data race prediction in practice. In windowing, the data race prediction problem aims to detect races within all subtraces of a fixed length \cite{cp2012,maxcausalmodels}. However, as we have shown in this paper, race definitions are not always well-defined on subtraces, which may result in missing certain short races or detecting races only under specific conditions, thereby losing predictive power beyond just event distance. The short race detection problem considered in this paper generalizes windowing by aiming to detect races within a fixed number of events while preserving the correctness of the underlying race definition. Sliding window algorithms~\cite{datar-sliding} extend the windowing paradigm by asking whether a property holds over the most recent $\winsz$ events of a stream. Our work on detecting short races can be seen as a specialization of this idea to the property of races, with the key distinction that we permit certain events outside the window to influence whether a race is detected within it.
}

\myparagraph{Sampling Based Race Detection}{
Sampling-based race detection \cite{marino2009literace,pacer-tool,uclock2024,racemob} reduces overhead by analyzing only a subset of events. Our \shortfasttrack algorithm shares the spirit of \cite{RPT2023}, which samples fixed-length subtraces to detect happens-before races. While \cite{RPT2023} focuses on which subtraces to analyze, our work studies the regime in which every subtrace is sampled and characterizes the inherent complexity of analyzing them. Conceptually, the short-race problem is a special case of subtrace-sampling approaches.

Although the literature on sampling-based race detection is extensive, it has largely targeted happens-before races. Comparatively little attention has been given to sampling for other predictive notions, precisely because designing a sampling scheme for these richer relations is technically difficult: they impose non-local constraints (e.g., synchronization preservation and feasibility of witness reorderings), so naïve subsampling can break soundness, while principled schemes risk prohibitive overhead. In this context, our method for detecting short sync-preserving races lays a foundation for future work along this avenue.

}

%% file: conclusions.tex
\section{Conclusions}
\seclabel{conc}

\sloppy
We introduce and analyze the problem of detecting short races in an observed trace. We present algorithms for detecting short happens-before and sync-preserving races that achieve similar or better running times and require less space compared to classical algorithms for these races. Our experiments highlight the effectiveness of these short race detection algorithms, demonstrating that they can significantly reduce overhead while still identifying many races. In the future, it would be valuable to explore short race detection algorithms for other types of races. Additionally, investigating the impact of such algorithms on randomized race detection, which often relies on identifying short races, would be an interesting avenue for further research.

%% file: appendix.tex
\newpage

\section{Proof Sketch of \thmref{algo2-correctness}}
\applabel{algo2-correctness}

We first prove that \algoref{algo2} correctly detects all happens-before races that spans by at most $k$ events. The key invariant of the algorithm is that every entry in each vector clock maintained by threads and locks records the window index of the most recent release event. Formally, for any thread vector clock $\CFT_t$, at any event $e$ processed by thread $t$, the entry $\CFT_t[t']$ stores the last release performed by thread $t'$ that happens-before $e$. If no such release exists within the window, the entry $\CFT_t[t']$ is set to $-1$. The same logic applies to the lock vector clock $\CFT_\lk$.

Similarly, for each read clock or write epoch, the entry corresponds to the window index of a read or write event within the window. Specifically, $\epch^{w}_x$ stores the window index of the last write to memory location $x$. If no such write event remains in the window, $\epch^{w}_x$ is set to $-1$. This convention also holds for all entries in the read vector clocks $\RFT_x$.

The window comparator $\leq_\window$ compares two window indices $a$ and $b$ with respect to the head position $h$, we have $a\leq_\window b$ if: (a) $a=-1$ or (b) $a\leq b\leq h$ or (c) $h< a\leq b$ or (d) $b \leq h < a$. This together with the way every vector clock gets updated ensures that $\RFT_x \not\cle_{\window} \CFT_t$ or $\epch^{w}_x \not\cle_{\window} \CFT_t$ holds if and only if there are pair of conflicting events in the window that are forms a happens-before races. The completes the correctness proof.

Now we prove the running time complexity and the space complexity analysis are correct. It is straightforward to see that the four handlers of acquire, release read and write events runs with the same running time complexity as those of \fasttrack. The post handler of read and write events also runs in constant time. Although the post-handler of release can access up to $\numthr + \numlk$ many entries, observe that the \textbf{post-handler} for a release event simply ``undoes'' the effect of previous events that led to the local time of the event being stored. By assigning the cost of each step in the \textbf{post-handler} to the acquire event that copied the local time of the release event being dropped, we see that taking the time for the \textbf{handler} of an acquire to be $O(2\numthr)$ accounts for all the time spent in the \textbf{post-handler}s for release events. Thus, the total running time is just $O(\trsz\numthr)$ which is the same as the running time of {\fasttrack}. The memory used by the algorithm is as follows. $O(\winsz)$ space is used to store the events in the window. We have a vector clock for each thread and the number of vector clocks associated with variables and locks is at most $\winsz$ and each entry of a vector timestamp is at most $\winsz$. Next, every element of in a set $S_i$ for a release event corresponds to a component of some vector clock associated with a thread or lock. Thus, the sum of the sizes of $S_i$ for all release events is also at most $O((\min(\winsz,\numlk) + \numthr)\numthr)$. Putting all this together, we see that the total space is $O(\winsz + (\min(\winsz,\numvar+\numlk)+\numthr)\numthr\log \winsz)$.

\section{Proof Sketch of \thmref{algo3-correctness}}
\applabel{algo3-correctness}

We first argue that \algoref{algo3} correctly solves the short sync-preserving race detection problem. To establish this, it is sufficient to demonstrate that when an event $e_2$ from the stream is being processed by the algorithm, for every conflicting event $e_1$ that is within $k$ events of $e_2$, \algoref{algo3} effectively decides for every event $e$ in the window whether $e\in \SPIdeal{\tr}(e_1, e_2)$ .Since \algoref{algo3} effectively simulates the \syncp algorithm, we first refer readers to the technical report \cite{SyncP2021} for detailed background. Recall that \algoref{algo3} computes $\TLClosure{\tr}(e)$ for every event $e$ and maintains them in their respective lists as long as $e$ remains within the window. As established in \cite{SyncP2021}, given $\TLClosure{\tr}(e_1)$ and $\TLClosure{\tr}(e_2)$, computing $\SPIdeal{\tr}({e_1,e_2})$ requires incrementally determining whether each release event $e$ should be included during the fixed point computation. This is done by comparing $\match{\tr}(e)$ with acquires that occur after $e$.

In our case, we need to decide whether each release event within the window should be included. If a release event $e$ has its matching acquire $\match{\tr}(e)$ still within the window, this case is naturally handled by \syncp, as \algoref{algo3} preserves the information of $\match{\tr}(e)$ along with all subsequent acquires, which must also be within the window due to the presence of $e$. When $\match{\tr}(e)$ moves out of the window, \algoref{algo3} retains the necessary information through the function \spill{}, enabling it to invoke \syncp to process the case properly. This ensures the soundness of the approach.

Now to prove that \algoref{algo3} has the same running time as \syncp(or a better running time when $k$ is sufficiently small), we need to show that the four additional functions do not increase the asymptotic complexity. It is straightforward to see that the functions \spill{} and \replay{} operates in constant time as the functions only remove the head node from a list, store the node elsewhere and put the node back to the list(which only happens when the list is empty). The function \allocate($e$) allocates data structures that \syncp would access and operate while processing $e$. For every vector clock allocated, \syncp would perform at least an $O(\numthr)$ operation so allocating a vector clock of size $\numthr$ does not change the running time. It is also straightforward to see that for all other data structures the allocation cost is constant and hence does not change running time either.  The argument for \deallocate is similar to that of \algoref{algo2}, the function merely ``reverts'' operations previously performed by \syncp, ensuring it does not introduce additional overhead. Finally, we note that the running time of \syncp on an execution $\tr$ can be analyzed by counting the number of visits to each node stored in the linked lists. With a similar argument as the analysis in \cite{SyncP2021}, we can show every node is visited at most $\winsz$ times. When $\winsz$ is sufficiently small, \algoref{algo3} establishes a better running time. When $\winsz$ is big (bigger than $\numvar\numthr^2$),for nodes that store information about memory access events, \algoref{algo3} visits them at most as frequently as \syncp, with the visits being bounded by $k$. However, the total cost of \algoref{algo3} when visiting an acquire node may double in the worst case compared to \syncp. This occurs when the acquire node is restored via \replay{} and is treated as a newly inserted node. Nevertheless, this does not affect the overall asymptotic complexity.

The space complexity can be analyzed in the following way: \algoref{algo3} keeps a vector clock for every lock and memory location, which in turn gives a space complexity of $O((\numlk+\numvar)\numthr)$. For a given window of size $\winsz$, it keeps $\winsz$ vector clocks, one for each event, which is $O(\winsz\numthr)$.  Every event $e$ in the window can contribute at most $\numthr^2\numlk$ many pointers, since there can be at most $\numvar\numlk\numthr^3$ many pointers, we have another factor of $O(\min(\winsz,\numvar\numthr)\numlk\numthr^2)$.